\documentclass[twocolumn,secnumarabic,amssymb, nobibnotes, aps, prd]{revtex4-2}
\def\BibTeX{{\rm B\kern-.05em{\sc i\kern-.025em b}\kern-.08em
    T\kern-.1667em\lower.7ex\hbox{E}\kern-.125emX}}
\usepackage{graphics}
\usepackage{color}
\usepackage{url}
\usepackage{subcaption}
\usepackage{latexsym}
\usepackage{amsfonts,amssymb,amsmath}
\usepackage[dvips]{graphicx}
\usepackage{epsfig}
\usepackage{verbatim}   
\usepackage{amsmath}
\usepackage{amssymb}
\usepackage{fancyhdr}
\usepackage[driver=pdftex]{geometry}

\usepackage[linesnumbered,ruled]{algorithm2e}

\usepackage{algpseudocode}
\usepackage[utf8]{inputenc}
\usepackage{amsthm}
\allowdisplaybreaks
\usepackage{mathrsfs}
\usepackage[normalem]{ulem}
\usepackage{array}
\usepackage{tikz}
\usepackage{quantikz}
\usepackage[hidelinks]{hyperref}
\usepackage[nameinlink]{cleveref}
\usepackage{booktabs}

\crefname{theorem}{thm.}{thm.}
\crefname{section}{sec.}{sec.}
\crefname{corollary}{cor.}{cor.}
\crefname{lemma}{lemma}{lemma}
\Crefname{theorem}{Thm.}{Thm.}
\Crefname{section}{Sec.}{Sec.}
\Crefname{figure}{Fig.}{Fig.}
\Crefname{equation}{Eqn.}{Eqn.}
\newtheorem{definition}{Definition}
\newtheorem{theorem}{Theorem}
\newtheorem{lemma}{Lemma}

\newtheorem{remark}{Remark}

\newcommand{\ev}{\mathbf{e}}
\newcommand{\bv}{\mathbf{b}}

\newcommand{\Am}{\mathbf{A}}

\newcommand{\Um}{\mathbf{U}}

\newcommand{\In}{\mathbf{I}}
\newcommand{\Id}{\mathbf{I}}

\newcommand{\Vm}{\mathbf{V}}

\newcommand{\Dm}{\mathbf{D}}

\newcommand{\Sp}{\sigma_+}
\newcommand{\Sm}{\sigma_-}

\newcommand{\ra}{\rangle}

\newcommand{\uv}{\mathbf{u}}

\newcommand{\Rr}{\mathbb{R}}
\newcommand{\Cr}{\mathbb{C}}

\newcommand{\Mm}{\mathbf{M}}
\newcommand{\fv}{\mathbf{f}}
\newcommand{\wv}{\mathbf{w}}
\newcommand{\vv}{\mathbf{v}}

\newcommand{\ab}{\textcolor{red}}

\newlength\myindent
\algdef{SE}[SUBALG]{Indent}{EndIndent}{}{\algorithmicend\ }%
\algtext*{Indent}
\algtext*{EndIndent}

\usepackage{mathtools}

\usepackage{booktabs}
\usepackage{thm-restate}

\usepackage{nicematrix}

\usepackage{url}

\begin{document}

\title{Efficient Quantum Access Model for Sparse Structured Matrices using Linear Combination of ``Things”}
\author{Abeynaya Gnanasekaran}%
\email[]{abeynaya.gnanasekaran@sri.com}
\affiliation{SRI International, Menlo Park, CA, USA}
\author{Amit Surana}%
\email[]{amit.surana@rtx.com}
\affiliation{RTX Technology Research Center (RTRC), East Hartford, CT, USA}

\begin{abstract}
We present a novel framework for Linear Combination of Unitaries (LCU)-style decomposition tailored to structured sparse matrices, which frequently arise in the numerical solution of partial differential equations (PDEs). 
While LCU is a foundational primitive in both variational and fault-tolerant quantum algorithms, conventional approaches based on the Pauli basis can require a number of terms that scales quadratically with matrix size. 
We introduce the Sigma basis, a compact set of simple, non-unitary operators that can better capture sparsity and structure, enabling decompositions with only polylogarithmic scaling in the number of terms. 
We develop both numerical and semi-analytical methods for computing Sigma basis decompositions of arbitrary matrices. Given this new basis is comprised of non-unitary operators, we leverage the concept of unitary completion to design efficient quantum circuits for evaluating observables in variational quantum algorithms and for constructing block encodings in fault-tolerant quantum algorithms.  
We compare our method to related techniques like unitary dilation, and demonstrate its effectiveness through several PDE examples, showing exponential improvements in decomposition size while retaining circuit efficiency.
\end{abstract}
\maketitle
\section{Introduction}

Linear Combination of Unitaries (LCU) is a versatile quantum algorithmic primitive that underlies a wide range of quantum algorithms, including those for Hamiltonian simulation, solving linear systems and  differential equations, implementing quantum walks, ground-state preparation, and optimization. It enables the decomposition of a non-unitary operator as a weighted sum of unitary operators which are typically tensor products of Pauli operators. Over the years, LCU has played a central role in the development of numerous quantum algorithms both in the context of variational quantum algorithms (VQAs) and fully fault-tolerant ones. 

VQAs are hybrid classical-quantum algorithms that have emerged as promising candidates to optimally utilize today's Noisy Intermediate Scale Quantum (NISQ) devices. 
Notable examples include, Variational Quantum Eigensolver (VQE) \cite{VQE} for ground-state energy estimation of a Hamiltonian and Variational Quantum Linear Solver (VQLS) \cite{VQLS} for solving system of linear equations. For instance, in VQE and VQLS, the Hamiltonian or system matrix is first decomposed via LCU into a sum of unitaries. Then, quantum subroutines—such as the Hadamard test—are used along with the parametrized quantum circuits (PQC) to measure the expectation value of each unitary term in LCU and the results are combined to evaluate the cost function.

LCU has also been used to construct block-encodings of non-unitary operators~\cite{berry2019qubitization}, which serve as powerful building blocks for many fault-tolerant quantum algorithms. Block-encoding enables one to evaluate the action of non-unitary operators on quantum states and plays a key role in a variety of contexts, including qubitization~\cite{low2019hamiltonian}, quantum signal processing~\cite{QSP}, and quantum singular value transformation (QSVT)~\cite{QSVT}.

A major bottleneck in LCU based methods is the number of terms $L$ in the LCU decomposition as the measurement cost often scales polynomially in the number of terms (e.g., $\mathcal{O}(L^2)$ in VQLS). 
Although Pauli-based decompositions yield low-depth circuits, the number of LCU terms can scale as polynomially in the size of the Hamiltonian or system matrix. To address this, LCU terms are often grouped into commuting subsets for simultaneous measurement~\cite{comm}, using classical heuristics based on graph coloring or clique cover. Additional techniques for measurement reduction include optimized sampling~\cite{samp}, classical shadows~\cite{shad}, and neural network tomography~\cite{tomo}.

For decomposing sparse and structured matrices, recent work has shown that replacing Pauli operators with a small set of simple, non-unitary operators $\mathbb{S} =\{\Id, \Sp=\ket{0} \bra{1}, \Sm=\ket{1}\bra{0}, \Sp\Sm = \ket{0}\bra{0}, \Sm\Sp=\ket{1}\bra{1}\}$, referred to as the Sigma basis, can significantly reduce the number of decomposition terms. Such matrices frequently arise in the numerical solution of partial differential equations (PDEs), which are central to many applications in science and engineering. For instance, \cite{liu2021variational} showed that for the matrix arising from the discretized Poisson equation, the number of Sigma-basis terms scales logarithmically with matrix size.
However, since Sigma basis consists of non-unitary operators, evaluating cost functions in VQAs requires specially designed observables and circuits. This limits the general applicability of the technique, as constructing observables for arbitrary tensor products of elements from $\mathbb{S}$ is non-trivial, and measurement typically requires access to all qubits, increasing overhead.

A more general and scalable approach was introduced in~\cite{gnanasekaran2024efficient}, where the authors used unitary completion to construct efficient quantum circuits to evaluate both global and local VQLS cost functions. Applied to the time-dependent heat equation, this method yields a number of decomposition terms that scale logarithmically with both spatial and temporal grid size as opposed to linear scaling with Pauli basis. Building on ~\cite{gnanasekaran2024efficient}, this paper makes several new contributions:
\begin{itemize}
  \item We develop numerical and semi-analytical methods for computing Sigma basis decompositions for arbitrary matrices.
  \item We present a general pseudocode and resource estimates for constructing efficient quantum circuits that implement unitary completion of non-unitary tensor-product operators from the Sigma basis.
  \item We construct Hadamard-style test circuits to evaluate observables commonly encountered in VQAs using Sigma basis decompositions.
  \item We propose a block-encoding framework for arbitrary operators expressed in the Sigma basis, with resource estimates for its use in fault-tolerant algorithms.
  \item We provide a rigorous comparison with related techniques, including unitary dilation and Pauli-based LCU decomposition.
  \item Finally, we demonstrate our approach on several PDE examples, showing significant numerical advantages over Pauli-based decompositions.
\end{itemize}

The rest of the paper is organized as follows. In~\Cref{sec: prelim}, we introduce the mathematical preliminaries. \Cref{sec: sigma} summarizes the Sigma basis decomposition and the concept of unitary completion. In~\Cref{subsec:Um}, we present the pseudocode for constructing the unitary completion operator, illustrate it with an example, and prove its correctness. \Cref{sec: methoddecomp} outlines numerical and semi-analytical methods for computing the Sigma basis decomposition of arbitrary matrices. In~\Cref{sec: apps}, we apply the decomposition to linear systems arising from discretizations of canonical linear PDEs and compare the results to Pauli basis LCU decompositions. In~\Cref{sec: uses_lcu}, we show how the non-unitary operators constructed using the Sigma basis can be used to compute observables in VQAs and to block-encode the system matrix for use in fault-tolerant quantum algorithms. In~\Cref{sec: discussion}, we analyze the connections between unitary completion and unitary dilation.  We conclude in~\Cref{sec: conc} with a summary and outline of future research directions.

\section{Preliminaries}\label{sec: prelim}
We will denote by $\mathbb{R}$ as the set of real numbers, $\mathbb{C}$ as the set of complex numbers,  small bold letters as vectors, capital bold letters as matrices/operators,  $\Am^*$ as the vector/matrix complex conjugate, $\Am^T$ as the vector/matrix transpose and $\Id_s$ as an Identity matrix of size $s \times s$. We will use standard braket notation in representing the quantum states. We use $C^n X$ to represent the \textit{n}-controlled Toffoli gate. With this notation, $C^1 X$ is the CNOT gate and $C^2 X$ is the Toffoli gate or CCNOT gate.

The rows of the matrix are numbered from $0$ to $2^n-1$ (0-indexed) for ease of mapping from decimal to binary system and vice-versa. The binary representation of a decimal number is given within $\{\dots\}$ and can be inferred from context. 

LCU decomposition of a matrix $\Am \in \Cr^{2^n\times 2^n}$ is defined as
\begin{equation}\label{eq:lcu}
\Am=\sum_{l=1}^{L}\alpha_l\Am_l,
\end{equation} 
where $\alpha_l\in \mathbb{C}$ and $\Am_l, l=1, \dots, L$  are unitary operators. A popular choice for building $\Am_l$ uses the Pauli operator basis, i.e. $\mathbb{P}=\{\sigma_x,\sigma_y,\sigma_z, \Id_2\}$ where $\sigma_x,\sigma_y,\sigma_z$ are Pauli-X, Pauli-Y and Pauli-Z single qubit operators,  so that
\begin{equation}
\Am_l=\sigma_0\otimes\dots\otimes\sigma_{n-1},\quad \sigma_l\in \mathbb{P}.
\end{equation}
When $\Am_l$ is instead chosen to be non-unitary operators, we  refer to the decomposition as Linear Combination of Non-Unitaries (LCNU). 

\section{Linear Combination of Non-unitaries Decomposition}
\label{sec: sigma}

In~\cite{liu2021variational,gnanasekaran2024efficient}, the authors developed a LCNU decomposition under a set of simple, albeit non-unitary, referred to as the Sigma basis.
\begin{definition}
The Sigma basis is the set
\begin{equation}\label{def:S}
\mathbb{S} =\{\Id_2, \Sp, \Sm, \Sp\Sm, \Sm\Sp\},
\end{equation}
where,
\begin{align*}
    \Sp &=\ket{0} \bra{1} = \begin{bmatrix}
        0 & 1 \\ 0 & 0
    \end{bmatrix} \quad
    \Sp\Sm = \ket{0}\bra{0} =\begin{bmatrix}
        1 & 0 \\ 0 & 0
    \end{bmatrix}
    \\
    \Sm  &=\ket{1}\bra{0} =\begin{bmatrix}
        0 & 0 \\ 1 & 0
    \end{bmatrix}  \quad
    \Sm\Sp =\ket{1}\bra{1} =\begin{bmatrix}
        0 & 0 \\ 0 & 1
    \end{bmatrix}
\end{align*}
\end{definition}
To address the non-unitarity of operators in the decomposition, the authors in~\cite{liu2021variational} designed specialized observables to compute expectation values with respect to each term in the decomposition. A more general approach was later introduced in~\cite{gnanasekaran2024efficient}, which bypasses the need for custom observables by using unitary completion—a framework that we extend in this work. In this section, we briefly review the key ideas from that approach.  

Given any arbitrary matrix $\Am \in \Cr^{2^n\times 2^n}$, its LCNU decomposition over the Sigma basis is defined by \Cref{eq:lcu} with
\[ \Am_l=\sigma_0\otimes\dots\otimes\sigma_{n-1},\quad \sigma_l\in \mathbb{S} \quad l=0, \dots, n-1.\]
For certain sparse and structured matrices, the number of terms $L$ in the LCNU decomposition can vary poly-logarithmically with problem size $N=2^n$, providing an exponential improvement over the standard LCU decomposition using Pauli basis~\cite{gnanasekaran2024efficient}. 
\begin{remark}\label{rem: example_spmat}
    For example, consider the decomposition of the following sparse matrix in Pauli and Sigma basis respectively,
\begin{align}
\begin{bmatrix}
        &&&1\\
        &&0&\\
        &0&&\\
        2&&
    \end{bmatrix}
= 
\begin{aligned}
    &0.75\,\sigma_x \otimes \sigma_x 
    -0.25i\,\sigma_x \otimes \sigma_y \\
    &- 0.25i\, \sigma_y \otimes \sigma_x 
    - 0.75\, \sigma_y \otimes \sigma_y
\end{aligned}
\end{align}
\begin{align}
    \begin{bmatrix}
        &&&1\\
        &&0&\\
        &0&&\\
        2&&
    \end{bmatrix} &= \Sp\otimes\Sp + 2\,\Sm \otimes \Sm.
\end{align}
The Pauli basis decomposition requires 4 terms for this example, while the Sigma basis uses only 2. More generally, the number of terms in the Sigma basis decomposition satisfies $L\leq \text{nnz}(A)$ where $\text{nnz}(A)$ are the number of non-zero entries in the matrix. For a more detailed discussion, see~\Cref{sec: methoddecomp}. For further details, see~\Cref{sec: methoddecomp}. Additionally, Pauli-based decompositions can introduce imaginary components through $\sigma_y$, even for real matrices. In contrast, the Sigma basis avoids this by construction; for complex matrices, complex values appear only in the coefficients, while the tensor products remain real. This can be advantageous in certain VQA cost evaluations.
\end{remark}

To efficiently compute expectation values of the form $\bra{\psi_1}\Am_l\ket{\psi_2}$, $\bra{\psi_1}\Am^*\ket{\psi_2}$ and $\bra{\psi_1}\Am^*\,\mathbf{M}\,\Am\ket{\psi_2}$, unitary operators of the form
\[
\Um_l = \begin{bmatrix}
    \Am_l & \star \\
    \star & \star 
\end{bmatrix},
\]
can be constructed. The construction of such a unitary operator is not unique and can be understood from the literature on block encoding, unitary dilation, projective measurements and unitary dynamics~\cite{nielsen2010quantum}. The key distinguishing feature of~\cite{gnanasekaran2024efficient} lies in its \textit{efficient} construction of $\Um_l$ for the Sigma basis using the concept of unitary completion. 

\begin{remark}
\label{rem:ul_bl}
    The unitary operator $U_l$ defined in~\cite{gnanasekaran2024efficient} was of the form $\begin{bmatrix}
        \star & \Am_l \\ \star & \star
    \end{bmatrix}$. We choose to redefine it slightly for easier interpretation and analogy with standard expressions of block encoding and unitary dilation. 
\end{remark}

Unitary completion is the process of extending a given set of orthonormal vectors to a full orthonormal basis of the complex vector space as given by~\Cref{def: completion}. Note that, as the basis set $\mathbb{S}$ is made up of real matrices, we can alternately use the term orthogonal completion. Using this, we embed partial orthonormal operators (such as the elements of $\mathbb{S}$ and their tensor products) into larger unitary systems of the form given in~\Cref{def: Ucomp}.  

\begin{definition}
\label{def: completion}
 Let $W\subset V\subseteq \mathbb{C}^N$ be complex-valued vector spaces. If $\mathbf{Q}: W \rightarrow V$ is a linear operator which preserves inner products, i.e., for any $\ket{w_1}$, $\ket{w_2}$ in $W \subset V$, $\bra{w_1}\mathbf{Q}^*\mathbf{Q}\ket{w_2} = \bra{w_1}w_2\ra$, then an unitary operator $\bar{\mathbf{Q}}: V\rightarrow V$ is its unitary completion if $\bar{\mathbf{Q}}$ spans the whole space $V$ and $\bar{\mathbf{Q}} \ket{w} =\mathbf{Q} \ket{w}\, \forall \ket{w} \in W$. Also, $\mathbf{Q}^c \coloneqq 
 \bar{\mathbf{Q}}-\mathbf{Q}$ is the orthogonal complement of $\mathbf{Q}$. Such a unitary operator $\bar{\mathbf{Q}}$ always exists (see Ex 2.67,~\cite{nielsen2010quantum}).
\end{definition}

\begin{definition}\label{def: Ucomp}
Consider an  operator $\Um_l$ associated with operator $\Am_l$ of the form
\begin{equation}
\label{eq: Um_l}
\Um_l = \begin{bmatrix}
\Am_l & \Am_l^c \\
\Am_l^c & \Am_l
\end{bmatrix},
\end{equation}
where, $\Am_l^c$ is the orthogonal complement of $\Am_l$.    
\end{definition}

By definition, the subspace spanned by the columns of the $\Am^c_l$ are orthogonal to the subspace spanned by the columns of $\Am_l$, i.e., $\Am_l^T\Am_l^c =0$ (see Lemma 1 in~\cite{gnanasekaran2024efficient} for proof). With this, it is straightforward to see that $\Um_l$ is in fact a unitary matrix. The application of $\Um_l$ on an ancilla system of the form $\ket{0}\ket{\psi}$ gives 
\begin{equation}
\label{eq: Umact}
\Um_l\ket{0} \ket{\psi} = \ket{0} \Am_l \ket{\psi} + \ket{1} \Am_l^c\ket{\psi}.
\end{equation} 

\begin{theorem}
\label{thm: Acomp}
If $\Am_l = \otimes_{k} \sigma_k$ where $\sigma_k \in \mathbb{S}$, then $\bar{\Am}_l\coloneqq\otimes_k \bar{\sigma}_k$ is the unitary completion of $\Am_l$, where $\bar{\sigma}_k = \sigma_x$ for $\sigma_k \in \{ \Sp, \Sm\}$ and $\bar{\sigma}_k = \Id$ for $\sigma_k \in \{ \Id, \Sp\Sm, \Sm\Sp\}$. 
\end{theorem}
\begin{proof}
See proof of Theorem 2 in~\cite{gnanasekaran2024efficient}.
\end{proof}

The unitary completion of the term $\Am_l = \otimes_{k} \sigma_k$ where $\sigma_k \in \mathbb{S}$ is given by~\Cref{thm: Acomp}. Note that the completion $\bar{\Am}_l$ has a simple closed-form expression in terms of Pauli matrices and hence, is straightforward to compute. The orthogonal complement $\Am_l^c = \bar{\Am}_l - \Am_l$ is generally harder to compute for any given $\Am_l$. However, one does not need to compute its expression explicitly in order to construct the circuit for implementing $\Um_l$. The construction of the quantum circuit is described in detail in~\Cref{subsec:Um} along with examples.

\subsection{Circuit construction for $
\Um_l$}
\label{subsec:Um}

The technique of unitary completion is a structured approach for block-encoding non-unitary operators $\Am_l$. The key feature of this technique lies in our ability to efficiently construct the circuit for the associated unitary operators $\Um_l$. We begin by providing a pseudo-code for the circuit construction and then prove its correctness in detail. 

Given a tensor-product over the Sigma basis, $\Am_l = \otimes_k \sigma_k$, the circuit for the corresponding unitary operator $\Um_l$ can be constructed by following~\Cref{circuitalg}. Note that the circuit construction only requires the computation of $\bar{\Am}_l$ and \textit{not} the complement $\Am_l^c$. Moreover, the circuit can be constructed with $n+1$ single qubit gates $\{\sigma_x, \Id\}$ and one $C^{m}X$ gate where $m \leq n$.

\begin{algorithm}
\caption{Pseudo-code for circuit construction of $\Um_l$ operator}\label{circuitalg}
\begin{algorithmic}[1]
\Require $\Am = \sigma_0 \otimes \sigma_1 \otimes \dots \otimes \sigma_{n-1}$ 
\Require $n+1$ qubit system with qubits $q_0, q_1, \dots, q_{n-1}$ and ancilla $a_0$.
\State Compute $\bar{\Am}_l = \otimes_n \bar{\sigma}_k$ where $\bar{\sigma}_k \in \{\sigma_x, I\}$
\State Apply the $n$ single-qubit gates that make up $\bar{\Am}_l$ 
\State Apply $X$ gate on $a_0$
\State Set $k = \sum_{p=0}^{n-1} [\sigma_p = \Id]$ where $[P]$ is defined as, \[[P]= \begin{cases}
    1,& \text{if true} \\
    0,& \text{if false}
\end{cases}\]
\State Construct a $C^{n-k} X$ gate as follows\;
\SetAlgoNoLine
\For{$i \gets 0$ \KwTo $n-1$}{
    \If{$\sigma_i \in \{\Sm, \Sm\Sp\}$}{
        Add closed control  on $q_i$\;
    }
    \If{$\sigma_i \in \{\Sp, \Sp\Sm\}$}{
        Add open control on $q_i$\;
    }
}
Add a X gate on the target $a_0$ 
\end{algorithmic}
\end{algorithm}

\textit{Example:} Let us walk through circuit construction using a simple example. Consider a three qubit system on $q_0, q_1, q_2$, with $\Am_l = \Sm \otimes \Id \otimes \Sp\Sm$. By adding an extra ancilla qubit $a_0$, $\Um_l$ can be constructed as follows and is shown in~\Cref{fig: U_eg}.
\begin{enumerate}
    \item Compute the unitary completion of $\Am_l$ as $\bar{\Am}_l = \overline{\Sm}\otimes \bar{\Id} \otimes \overline{\sigma_+\sigma}_- = \sigma_x \otimes \Id \otimes \Id$. 
    \item Construct the single-qubit gates $X$ on $q_0$ and $\Id$ on $q_1, q_2$.
    \item Apply $X$ gate on $a_0$.
    \item Compute $k$ as the number of factors in representation of $\Am_l$ that are Identity operations. In this case, $k=1$ as the second term in the tensor product is $\Id$.
    \item Construct a $C^2 X$ gate (CCNOT gate) by adding a closed control on $q_0$, open control on $q_2$ and a $X$ gate on target $a_0$.
\end{enumerate}

\begin{figure}[thbp]
    \centering
    \begin{quantikz}[wire types = {q, q, q, q}, transparent]
\lstick{$a_0$} &  \gate{X} & \targ{}  &  \\
\lstick{$q_0$} & \gate{X}&  \ctrl{-1} &\\
\lstick{$q_1$}  &  \gate{I}        &     & \\
\lstick{$q_2$}  &  \gate{I}        & \octrl{-2}  &  
\end{quantikz} 
    \caption{Circuit for $\Um_l$ corresponding to the term $\Am_l = \Sm \otimes \Id \otimes \Sp\Sm$.}
    \label{fig: U_eg}
\end{figure}
We have successfully constructed the circuit in five simple steps. Readers not interested in the proof of correctness of~\Cref{circuitalg} may skip the following subsection and continue to~\Cref{sec: methoddecomp}, where we discuss techniques for computing Sigma basis decompositions of arbitrary matrices.

\subsection{Correctness of construction}
\label{subsec: corr_cons}

The cost of constructing $\Um_l$, as defined in~\Cref{def: Ucomp}, is established in~\Cref{thm: circuit_U}. Before presenting the proof of this, we introduce a useful subroutine in~\Cref{lem:toff}. Although some of the material overlaps with~\cite{gnanasekaran2024efficient}, we include it here for completeness and to ensure readability. Moreover, since the definition of $\Um_l$ varies slightly as noted in~\Cref{rem:ul_bl}, a detailed treatment of~\Cref{thm: circuit_U} is provided for clarity.

\begin{lemma}\label{lem:toff}
    A $C^n X$ gate corresponds to a permutation matrix of size $2^{n+1} \times 2^{n+1}$ that permutes two rows whose binary representations differ by a single bit.
\end{lemma}
\begin{proof}
    $C^n X$ gate has $n$ control qubits and 1 target qubit by definition. WLOG, assume that the first $n$ qubits are control bits (with closed control) and NOT gate on the target $n+1$ qubit. This corresponds to a permutation matrix of the form
    \[
    P = \begin{bmatrix}
        \Id & \\
        & \sigma_x
    \end{bmatrix}.
    \]
    Applying $P$ on a vector (or matrix), permutes the last two rows of the said vector (or matrix). That is, it permutes two rows $r_1$ and $r_2$ where
    \begin{align*}
         r_1 &= \{ 11\dots 10\} = \sum_{p=0}^{n-1}2^{n-p}, \\
    r_2 &= \{11\dots 11\} = \sum_{p=0}^{n-1}2^{n-p} + 1.
    \end{align*}
    The terms in the $\{\dots\}$ are the binary representations. Note that the binary representations differ only in a single bit (which corresponds to the target qubit). 
\end{proof}

\begin{remark}\label{rem:control_op}
    A $C^n X$ gate operating on target qubit $k$ permutes two rows whose binary representation differ in the $k$-th bit. The control operation on the remaining $n$ qubits (open or closed control) determines the exact rows being permuted. If there is an open control (closed control) on $p \neq k$-th qubit, then the corresponding bit is 0 (1) in the binary representation of the rows. 
\end{remark}

The truth table for different versions (open/ closed control) of the CCNOT (or $C^2X$ gate) are given in~\Cref{app: ccnot}. They can be used as examples to understand the above remark.

\begin{restatable}{corollary}{kbitsswap} \label{cor:kbits_swap}
    $2k + 1$ $C^n X$ gates are needed to permute two rows in a $2^{n+1}\times 2^{n+1}$ matrix that differ by $k+1$ bits in their binary representation.
\end{restatable}
\begin{proof}
    See~\Cref{app:pf_cor1} for proof.
\end{proof}

\begin{lemma}
\label{lem: Arow}
    If $\Am_l = \otimes_p \sigma_p$, where $\sigma_p \in \mathbb{S}$, $r = \sum_{p=0}^{n-1} 2^{n-1-p} q_r(p)$, and $c=\sum_{p=0}^{n-1}2^{n-1-p} q_c(p)$, then \[
    \Am_l(r,c)= 1 \iff   \sigma_p(q_r(p),q_c(p)) = 1\quad  \forall p,\]
    where $q_r(p), q_c(p)\in \{0,1\}$.
\end{lemma}
\begin{proof}
    Let us prove the forward direction by induction on $n$. The result is trivially true for $n=1$, since $\Am_l^{(1)} \coloneqq \Am_l = \sigma_0$ and there is a one-to-one correspondence between the rows/ columns of $\Am_l$ and $\sigma_0$. Assume that the result holds true for $n-1$, i.e., for $
    \Am_l^{(n-1)} \coloneqq \otimes_{p=0}^{n-2} \sigma_p$,
    Then for $n$,  
    \begin{align*}
    \Am_l^{(n)} &= \sigma_{0} \otimes  \sigma_{1} \dots \otimes \sigma_{n-1}
    = \sigma_{0} \otimes \Am_l^{(n-1)}.
    \end{align*}
    The following relations hold:
    \begin{align*}
        \sigma_{0}(0,0)=1 &\implies \Am_l^{(n)} = \begin{bmatrix}
        \Am_l^{(n-1)} & \star \\ \star & \star 
    \end{bmatrix}, \\
    \sigma_{0}(0,1)=1 &\implies \Am_l^{(n)} = \begin{bmatrix}
        \star & \Am_l^{(n-1)} \\ \star & \star 
    \end{bmatrix}, \\
    \sigma_{0}(0,1)=1 &\implies \Am_l^{(n)} = \begin{bmatrix}
        \star & \star \\ \Am_l^{(n-1)} & \star 
    \end{bmatrix}, \\
    \sigma_{0}(1,1)=1 &\implies \Am_l^{(n)} = \begin{bmatrix}
        \star & \star \\ \star & \Am_l^{(n-1)} 
    \end{bmatrix}. 
    \end{align*}
    Therefore, $\Am_l^{(n)}(r,c)=1$ implies that $\sigma_{0}(q_r(0), q_c(0))=1$ and $\Am_l^{(n-1)}(r',c')=1$ where
    \begin{align*}
        r &= 2^{n-1}q_r(0) + r'\\
        c &= 2^{n-1}q_c(0) + c'.
    \end{align*}
    Hence proved by using the induction hypothesis. The other direction of the if and only if statement can be similarly proved using induction and we leave it as an exercise for the readers.
\end{proof}

\begin{theorem}
\label{thm: circuit_U}
${\Um}_l$  as defined in (\ref{eq: Um_l}) can be implemented using at most $n+1$ single qubit gates and a $C^m X$ gate where $m \leq n$.
\end{theorem}
\begin{proof}
${\Um}_l$ can be written as a product of two unitary matrices ${\Um}_{l,1}{\Um}_{l,2}$, i.e. ${\Um}_l={\Um}_{l,1} {\Um}_{l,2}$,   such that 
\begin{align*}
{\Um}_{l,2} &= \sigma_x \otimes \bar{\Am}_l = \begin{bmatrix}
 & \bar{\Am}_l \\
\bar{\Am}_l &  \\
 \end{bmatrix},  \\
 {\Um}_{l,1} &= {\Um}_l {\Um}_{l,2}^T =  \begin{bmatrix}
\Am_l & \Am^c_l \\
\Am^c_l & \Am_l
\end{bmatrix} \begin{bmatrix}
 & \bar{\Am}^T_l \\
\bar{\Am}^T_l &  \\
 \end{bmatrix}  \\
  &=
 \begin{bmatrix}
\Am^c_l (\Am^c)^T_l & \Am_l \Am^T_l \\
\Am_l \Am^T_l & \Am^c_l (\Am^c)^T_l \\
\end{bmatrix} \, (\because \text{by 
 definition})\\
 &=
 \begin{bmatrix}
\Id - \Am_l \Am^T_l & \Am_l \Am^T_l \\
\Am_l \Am^T_l & \Id - \Am_l \Am^T_l
\end{bmatrix}.  
\end{align*}
Using \Cref{thm: Acomp}, ${\Um}_{l,2}$ only involves tensor products of $\sigma_x, \Id$, and thus can be implemented efficiently using single qubit gates. We require at most $n$ single qubit gates to implement $\bar{\Am}_l$ and and one $X$ gate on the ancilla bit.

Note that $\Am_l \Am_l^T$ is a binary diagonal matrix with 1's and 0's at the diagonal
\begin{align*}
\Am_l \Am_l^T &= \otimes_p \sigma_p \sigma_p^T, \quad \sigma_p \sigma_p^T &\in \{\Sp\Sm, \Sm\Sp, \Id \},
\end{align*}
as each term $\sigma_p\sigma_p^T$ is a diagonal matrix. Therefore, ${\Um}_{l,1}$ is a $2^{n+1} \times 2^{n+1}$ permutation matrix. Any permutation matrix can be implemented using only Toffoli gates~\cite{nielsen2010quantum}.  We prove that, only a \textit{single} $C^n X$ gate is required to implement ${\Um}_{l,1}$. 

$\Am_l \Am_l^T$ can have one or more non-zero rows. Let row $r$ of $\Am_l \Am_l^T$ be non-zero, where
\[r = \{q_p(0)\,\dots\,q_p(n-1)\} = \sum_{p=0}^{n-1} 2^{n-1-p} q(p),\]
and $q(p) \in \{0, 1\}$. Using~\Cref{lem: Arow}, we have
\begin{align*}
(\Am_l \Am_l^T)(r,r) &= 1 \quad \text{iff } (\sigma_p \sigma_p^T)(q(p), q(p)) = 1, \forall p, 
\end{align*}
and so
\begin{align}\label{eq:Urows}
\Um_{l,1}(r,r)&= \Um_{l,1}(2^n+r ,2^n+r)=0 \nonumber\\
\Um_{l,1}(r, 2^n+r) &= \Um_{l,1}(2^n+r , r)\hspace{22pt}=1.
\end{align}
Thus, $\Um_{l,1}$ can be constructed by permuting rows $r$ and $2^n + r$ of $\Id_{2^{n+1}}$. Consider, two cases:
\paragraph*{Case 1: $\Am_l \Am_l^T$ has a single non-zero row} 
Using ~\Cref{lem:toff}, the two rows can be permuted using a $C^n X$ gate.
\paragraph*{Case 2: $\Am_l \Am_l^T$ has multiple non-zero rows} 
This is possible only when one or more of the factors  $\sigma_p = \Id$ (using~\Cref{lem: Arow}). WLOG, assume that one of the factors, say $\sigma_{n-1} =\Id$. Then the two non-zero rows are
\begin{align*}
    r_1 &= \{q_{r}(0)\,\dots\, q_{r}(n-2)\,0\} = \sum_{p=0}^{n-2} 2^{n-1-p}q_{r}(p), \\
    r_2 &=\{q_{r}(0)\,\dots\, q_{r}(n-2)\,1\} = r_1 + 1.
\end{align*}
Naively, two $C^n X$ gates are needed to permute $r_1$ with $2^n + r_1$ and $r_2$ with $2^n + r_2$. However, since the binary representations of $r_1$ and $r_2$ differ only in the $n$-th bit, the two $C^n X$ gates only differ in the control operation on the $n$-th qubit (by~\Cref{rem:control_op}). Hence, the two $C^n X$ gates can be combined into a singe $C^{n-1} X$ gate with the qubits $q_0, q_1, \dots, q_{n-2}$ acting as control qubits (no control operation on $q_{n-1}$) and target on $a_0$ (see example below). Note, that there is no operation on the $q_{n-1}$. In general, if $k$ factors are $\Id$, then there are $2^k$ non-zero rows and we would require a $C^{n-k} X$ gate with $n-k$ control qubits and one target qubit on $a_0$.
\end{proof}

\paragraph*{How to choose the control operations in $C^n X$ gate?} The control operations of $C^n X$ gate are determined using the tensor decomposition of $\Am_l = \otimes_p \sigma_p$ as follows. Using~\Cref{lem: Arow}, the non-zero rows of $\Am_l\Am_l^T$ can be identified:
\[
    (\Am_l \Am_l^T)(r,r) = 1 \quad \text{iff } (\sigma_p \sigma_p^T)(q(p),q(p)) = 1, \forall p,
\]
where $r = \sum_{p=0}^{n-1}2^{n-1-p}q(p)$. We can use~\Cref{rem:control_op} to determine the control operations:
\begin{align*}
\sigma_p \in \{\Sp, \Sp\Sm\} &\implies 
    q(p)=0 \\
    &\implies \sigma_p \sigma_p^T = \begin{bmatrix}
1 & 0\\
0 & \star
\end{bmatrix} \\
 &\implies \text{open-control on $q_p$}, \\
\sigma_p \in \{\Sm, \Sm\Sp\} &\implies q(p)=1 \\ &\implies \sigma_p \sigma_p^T = \begin{bmatrix}
\star & 0\\
0 & 1
\end{bmatrix} \\
& \implies \text{closed-control on $q_p$},
\end{align*}
where $q_p$ is the corresponding qubit index. 
If there are $k$ factors $\sigma_p=\Id$, then there will be $2^k$ non-zeros rows in  $\Am_l\Am_l^T$. In this case, no control operation is necessary on the corresponding $k$ qubits and the $C^n X$ gate can be simplified to a $C^{n-k} X$ gate. The target is always on the ancilla bit as we want to permute rows that differ in the most significant bit (see~\Cref{eq:Urows}). The template circuit for $\Um_l$ is shown in~\Cref{fig:Uml_circuit}.

\begin{remark}
    The proof of~\Cref{thm: circuit_U} also provides the circuit construction of $\Um_l$ as described in~\Cref{circuitalg}. Steps 1-3 of the algorithm implement $\mathbf{U}_{l,2}$ while steps 4-5 implement $\mathbf{U}_{l,1}$.
\end{remark}

\begin{figure}
    \centering
\begin{quantikz}[wire types = { q, q, q, n, q, q}, transparent]
\lstick{$a_0$ } & &         \gate{X} \wire[d][1]{q}\gategroup[6,steps=2,style={dashed,rounded
corners,fill=blue!20, inner
xsep=2pt},background,  label style={label
position=below,anchor=north,yshift=-0.2cm}]{{\sc ${\Um}_l$}}  & \targ{} &&  \\
\lstick{$q_0$ }& & \gate[5]{\bar{\Am}_l} &  \octrl{-1}  &  & \\
\lstick{$q_1$ } &  & &\ctrl{-1}  & & \\
\lstick{\vdots} && & \vdots&   &\\
\lstick{$q_{n-2}$ }&  &&\octrl{1}  && \\
\lstick{$q_{n-1}$ }& &  &\ctrl{-1} & & 
\end{quantikz}
\caption{Circuit for $\Um_l$ for an given $\Am_l = \otimes_i \sigma_i$ for $\sigma_i \in \mathbb{S}$. The completion operator is $\overline{\Am}_l=\otimes_i \bar{\sigma}_i$, where $\bar{\sigma}_i \in \{\Id, \sigma_x\}$ and is defined as in~\Cref{thm: Acomp}.}
    \label{fig:Uml_circuit}
\end{figure}

\textit{Revisiting our example:} Consider the example three qubit system used in~\Cref{subsec:Um}, $\Am_l = \Sm \otimes \Id \otimes \Sp\Sm$. Steps 1-3 of the algorithm implement $\mathbf{U}_{l,2}$ trivially. Let's expand on steps 4-5. There are two non-zero rows $r_1, r_2$ in $\Am_l \Am_l^T$ where
$
r_1 = \{0100\} \text{ and } r_2 = \{0110\}. 
$ Note that most significant bit corresponds to the ancilla qubit. The binary representation of these rows differ in the $q_1$ bit as the corresponding factor in $\Am_l$ is $\Id$ and corresponds to case 2 in the proof of~\Cref{thm: circuit_U}. In order to construct $\Um_{l,1}$ (on 4-qubit), two $C^3 X$ gates are needed to permute
\begin{align*}
    r_1 = \{0100\} &\leftrightarrow r_1' = \{1100\}, \\
    r_2 = \{0110\} &\leftrightarrow r_2' = \{1110\}.
\end{align*} 
The first permutation ($r_1 \leftrightarrow r_1'$) can be done using a $C^3 X$ with closed control on $q_0$ and open control on $q_1, q_2$. The second permutation ($r_2 \leftrightarrow r_2'$) can be done with another $C^3 X$ with closed control on $q_0,q_1$ and open control on $q_2$. The target is always on the ancilla  $a_0$. However, as the two $C^3 X$ gates differ only in the control operation on $q_1$, they can be effectively represented by a single $C^2 X$ gate as shown in~\Cref{fig: ex_toff}. This can easily verified by constructing truth tables as exemplified in~\Cref{app: ccnot}.
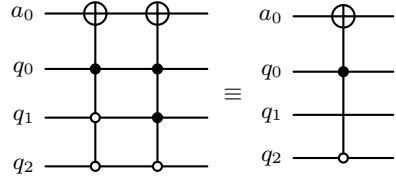
\begin{figure}[htbp]
    \centering
    \begin{quantikz}[wire types = {q, q, q, q}, transparent]
\lstick{$a_0$} & \targ{}  & \targ{} & \\
\lstick{$q_0$} &  \ctrl{-1} & \ctrl{-1} &\\
\lstick{$q_1$}   & \octrl{-1}& \ctrl{-1} &\\
\lstick{$q_2$}   & \octrl{-1}  &  \octrl{-1} &
\end{quantikz} $\equiv $
\begin{quantikz}[wire types = {q, q, q, q}, transparent]
\lstick{$a_0$} & \targ{}  & \\
\lstick{$q_0$} &  \ctrl{-1} &\\
\lstick{$q_1$}   &  &\\
\lstick{$q_2$}   & \octrl{-2}   &
\end{quantikz} 
    \caption{Two $C^3X$ gates that differ only in the control operation on one qubit can be effectively combined into a single $C^2 X$ gate.}
    \label{fig: ex_toff}
\end{figure}

\section{Techniques for computing sigma decomposition of a complex matrix $A$}\label{sec: methoddecomp}

For sparse matrices with structured nonzero patterns, decomposing over the Sigma basis can be more advantageous than using the standard Pauli basis. Since Pauli matrices are Hermitian, the number of LCU terms can exceed the number of nonzero entries $\text{nnz}(A)$, as illustrated in~\Cref{rem: example_spmat}.  In contrast, the Sigma basis decomposition guarantees an upper bound of $\text{nnz}(A)$. For sparse matrices where $\text{nnz}(A)\ll N$, we recommend using the numerical approach to compute the decomposition. 

\textbf{Numerical Approach:} Express $\Am$ as a sum of matrices $\tilde{\Am}_i\in \mathbb{C}^{2^n\times 2^n}$ each with a single non-zero entry, i.e.,
\begin{equation}\label{eq:sigentry}
\Am=\sum_{i=1}^{\text{nnz}(A)}\tilde{\Am}_i.
\end{equation} 
Then by invoking the~\Cref{thm:sigma_decomp}, $\tilde{\Am}_i$ can be represented as a single tensor product term over the Sigma basis. Note that since $\Id\in \mathbb{S}$ is not used in the tensor product, this approach may not produce a decomposition with the minimum number of terms. 

\begin{theorem}\label{thm:sigma_decomp}
Consider $\Am \in \mathbb{C}^{2^n \times 2^n}$ with a single nonzero entry in position $(r,c)$ with $r = \sum_{p=0}^{n-1} 2^{n-1-p} q_r(p)$ and $c = \sum_{p=0}^{n-1} 2^{n-1-p} q_c(p)$.  Then the LCNU decomposition of $\Am$ is given by $\Am = \Am(r,c) \left(\otimes_p \sigma (q_r(p), q_c(p))\right)$ where
\begin{align*}
\sigma (q_r(p), q_c(p)) &= \begin{cases}
  \Sp\Sm  &\hspace{-3pt}, q_r(p)=0,q_c(p)=0 \\
   \Sp &\hspace{-3pt}, q_r(p)=0,q_c(p)=1 \\
   \Sm &\hspace{-3pt}, q_r(p)=1,q_c(p)=0\\
  \Sm\Sp &\hspace{-3pt}, q_r(p)=1,q_c(p)=1
\end{cases}
\end{align*} 
\end{theorem}
\begin{proof}
This is a direct implication of~\Cref{lem: Arow} and is only restated here for clarity in the given context. 
\end{proof}

For a full-rank matrix, $\text{nnz}(A)\geq N$. In order to get an ideal $\mathcal{O}(\text{poly}\log N)$ number of LCNU terms, the matrix typically needs to have a recursive/ telescoping pattern with \textit{few} non-zero entries ($\ll N$) that stray away from the recursive structure. This is illustrated in~\Cref{fig: recursive}. There is no general recipe for determining such patterns and one has to proceed on a case-by-case basis based on the structure/sparsity of the given matrix. We provide a walkthrough of the semi-analytical approach, combining the recursive structure identification and the numerical approach next. 

\begin{figure*}
    \centering
\[
\begin{array}{ccc}
\Am^{(2)} =
\begin{bNiceArray}{cw{c}{1em}cw{c}{1em}}[
    margin,
]
  \star & \star \\
  \star & \star 
\end{bNiceArray}
 &  \quad
\Am^{(4)} =
\begin{bNiceArray}{cw{c}{1em}cw{c}{1em}cw{c}{1em}cw{c}{1em}}[
    margin,
    vlines={3},   
]
  \Block{2-2}{\Am^{(2)}} &  &  &  \\
                       &  & \star &  \\ 
  \hline
   & \star & 
  \Block{2-2}{\Am^{(2)}} & \\
   &  &          &  
\end{bNiceArray}
& \quad \Am^{(8)} =
\begin{bNiceArray}{cw{c}{1em}cw{c}{1em}cw{c}{1em}cw{c}{1em}cw{c}{1em}cw{c}{1em}cw{c}{1em}cw{c}{1em}}[
    margin,
    vlines={5},   
]
  \Block{4-4}{\Am^{(4)}} &  & & &  & & & \\
                       &  & & &  & & & \\
                      &  & & &  & & & \\ 
                     &  & & & \star & & & \\ 
                       
  \hline
   & & & \star & 
  \Block{4-4}{\Am^{(4)}} & & & \\
                       &  & & &  & & & \\
                       &  & & &  & & & \\
                       &  & & &  & & & \\
\end{bNiceArray} \\ \\
L^{(2)} \leq 4 & \quad L^{(4)} = L^{(2)} + 2  & \quad L^{(8)} = L^{(4)}+2
\end{array}
\]
    \caption{Illustration of a recursive/ telescoping matrix structure that leads to a $\mathcal{O}(\log N)$ number of LCNU terms. At each level of the recursion, there are only 2 non-zero entries that stray outside of the recursive pattern and can be handled efficiently with the numerical approach. $A^{(n)}$ and $L^{(n)}$ indicate the matrix and number of LCNU terms for matrix size $n$. }
    \label{fig: recursive}
\end{figure*}

\textbf{Semi-analytical Approach:} 
Consider a linear system of ordinary differential equations (ODEs),
\begin{equation*}
\frac{d\uv}{dt} = \Am^\prime\uv+\bv^\prime,
\end{equation*}
where, $\uv(t)\in \Rr^{n_x}, t\in[0,T]$, $T\in \Rr$ is the period of integration, $\Am^\prime\in \Rr^{n_x\times n_x}$ is the system matrix, $\uv(0)=\uv_0$ is the initial condition, and $\bv^\prime\in \Rr^{n_x}$ is a constant forcing term. Using explicit Euler time discretization with step size $\Delta t$, the above ODEs can be expressed as a system of difference equations 
\begin{equation*}
\uv_{k+1} =(\Id_{n_x}+\Delta t\Am^\prime)\uv_k+\Delta t\bv^\prime, k=1,\cdots,n_t-1,
\end{equation*}
where $\uv_k=\uv((k-1)\Delta t)$ with $\Delta t=T/(n_t-1)$. The above system can be represented as an extended system of linear equations
\begin{equation*}
\Am \uv = \bv,
\end{equation*}
where
\begin{align}\label{eq: full_coeffod}
    \uv &= \begin{bmatrix} 
\uv_1^T & \uv_2^T & \dots & \uv^T_{n_t}
\end{bmatrix}^T, \nonumber\\
\bv &= \begin{bmatrix}
\uv_0^T & \Delta t(\bv^\prime)^T & \dots & \Delta t (\bv^\prime)^T
\end{bmatrix}^T, \nonumber\\
\Am &= \begin{bmatrix*}[r]
\Id_{n_x} & & & 0\\
-\Id_{n_x} & \Id_{n_x} & & \\
& \ddots & \ddots & \\
0 & & -\Id_{n_x} & \Id_{n_x}
\end{bmatrix*}- \Delta t \begin{bmatrix*}[r]
0 & & & 0\\
\Am^\prime & 0 & & \\
& \ddots & \ddots & \\
0 & & \Am^\prime & 0
\end{bmatrix*} \notag\\
&= \Am_1 - \Delta t \Am_2.
\end{align}

For simplicity, assume $n_x= 2^s$ and $n_t = 2^t$. Then $\Am_1$ corresponds to $s+t$ qubits and can be written as
\begin{align*}
 \Am_1 \coloneqq \Am_1^{(s+t)} &= \begin{bmatrix}
 \Am_1^{(s+t-1)} & 0 \\
 \Dm_1^{(s+t-1)} & \Am_1^{(s+t-1)}
 \end{bmatrix}, \\
\Dm_1^{(s+t-1)} &= \begin{bmatrix}
0 & \dots & -\Id_{n_x} \\
\vdots & \ddots & \vdots \\
0 & \dots & 0
\end{bmatrix}\\ &= -\underbrace{\Sp \otimes \dots \otimes \Sp}_{t-1 \text{ times}} \otimes \Id_{n_x}.
\end{align*} 
The expression for $\Dm_1^{(s+t-1)}$ is obtained using the numerical approach. Then
\begin{align*}
\Am_1^{(s+t)} &= \Id_2 \otimes \Am_1^{(s+t-1)} + \Sm \otimes \Dm_1^{(s+t-1)}, \\
\Am_1^{(s+1)} &= \begin{bmatrix*}[r]
\Id_{n_x} & 0\\
-\Id_{n_x} & \Id_{n_x}
\end{bmatrix*} = \Id_2 \otimes \Id_{n_x} - \Sm \otimes \Id_{n_x}. 
\end{align*}
With this recursive relation, $\Am_1$ can be written using $\log n_t + 1$ terms. To obtain decomposition of $\Am_2$, we equivalently write it as
\begin{align}
\Am_2=\tilde{\Id}\otimes \Am^\prime,
\end{align}
where
\begin{align*}
\tilde{\Id}=\begin{bmatrix*}[r]
0 & & & 0\\
1 & 0 & & \\
& \ddots & \ddots & \\
0 & & 1 & 0
\end{bmatrix*}\in \Rr^{n_t\times n_t}.
\end{align*}
Applying a similar recursive procedure to $\tilde{\Id}$ as for $\Am_1$ above, we can express it as
\begin{align*}
\tilde{\Id} \coloneqq \tilde{\Id}^{(t)} &= \begin{bmatrix}
 \tilde{\Id}^{(t-1)} & 0 \\
 \tilde{\Dm}^{(t-1)} & \tilde{\Id}^{(t-1)}
 \end{bmatrix}, \\
\tilde{\Dm}^{(t-1)} &= \begin{bmatrix}
0 & \dots & 1\\
\vdots & \ddots & \vdots \\
0 & \dots & 0
\end{bmatrix}= \underbrace{\Sp \otimes \dots \otimes \Sp}_{t-1 \text{ times}}.
\end{align*} 
Then
\begin{align*}
\tilde{\Id}^{(t)} &= \Id_2 \otimes \tilde{\Id}^{(t-1)} + \Sm \otimes \tilde{\Dm}^{(t-1)}, \quad \tilde{\Id}^{(1)} = \Sm. 
\end{align*}
Finally, depending on the exact structure of $\Am'$, we can similarly apply a semi-analytical or numerical approach to compute its decomposition. Combining it with the decomposition of $\Am_1$ and $\tilde{\Id}$ as derived above, we  obtain the complete decomposition of $\Am$ with $L \leq (1+\text{nnz}(\Am'))(\log n_t + 1)$.

\section{Applications to Partial Differential Equations}\label{sec: apps}

\begin{figure*}[htbp]
\centering
\begin{subfigure}[b]{0.323\textwidth}
\includegraphics[trim={0 0 0 0},clip,width=\linewidth]{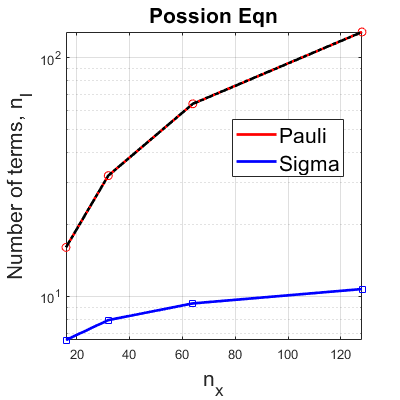}
\caption{Possion PDE}
\end{subfigure}
\begin{subfigure}[b]{0.32\textwidth}
\includegraphics[width=\linewidth]{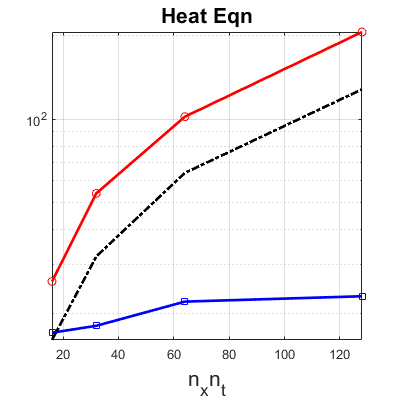}
\caption{Heat PDE}
\end{subfigure}
\begin{subfigure}[b]{0.32\textwidth}
\includegraphics[width=\linewidth]{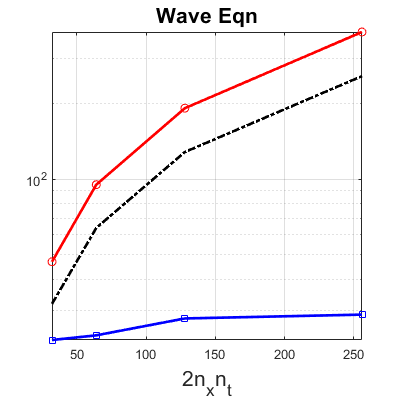}
\caption{Wave PDE}
\end{subfigure}
    
\caption{Comparison of number of terms in the Sigma basis and Pauli basis decomposition for matrices arising from discretization of 1D Possion, heat and wave PDEs. The black dotted line is $y=x$. For Possion PDE the number of Pauli terms grow linearly with $n_x$, while for the heat and wave PDE, they grow faster than linear. For Possion PDE we used  $n_x=16,32,64,128$, while for the heat and wave PDE we used $n_x(n_t)=
4(4), 4(8), 8(8), 8(16)$.} \label{fig: lcucomp}
\end{figure*}

It is well known that the standard LCU decomposition using Pauli basis is \textit{not} efficient for matrices arising from the discretization of linear PDEs.  Typically, these matrices are sparse and have a structued pattern of non-zero entries. Our LCNU technique using Sigma basis is best suited for such problems. 

We illustrate the applicability of our method by considering concrete examples of Poisson (elliptic), heat (parabolic), and wave (hyperbolic) equations in 1D.
For these examples, we apply the semi-analytical approach described in~\Cref{sec: methoddecomp} to obtain the LCNU decomposition with poly-logarithmic number of terms. We also compare our method with the standard Pauli basis LCU computed numerically using the matrix splicing method \cite{lcunum}. These comparisons indicate that the Sigma basis decomposition is \textit{exponentially} more compact/efficient compared to the Pauli basis decomposition for matrices arising from PDE discretizations. We also discuss generalization of these results for PDEs in higher dimensions.

\subsection{Elliptic PDE: Poisson Equation}\label{sec: ellipPDE}

Consider the 1D Poisson equation  defined over the domain $\Omega=[0,\, l]$ with  Dirichlet boundary conditions,
\begin{align*}
&-\frac{\partial^2u(x)}{\partial x^2}=f(x), \quad x\in \Omega\\
& u(0)=0,\quad u(l)=0.
\end{align*}
Using a second-order finite differing scheme, the discretized solution $\uv=(u(x_1),\cdots,u(x_{n_x}))^T$ over a finite grid $x_i=i\Delta x,i=1,\cdots,n_x$ with $\Delta x=l/(n_x+1)$, can be expressed as a linear system $\Am_e \uv=\bv$, where
\begin{align}\label{eq:possionA}
\Am_e &= \begin{bmatrix*}[r]
 2 & -1 &  &  & 0 \\
       -1 & 2 & -1 &  &  \\
       & \ddots & \ddots & \ddots &  \\
       &  & \ddots  & 2 & -1 \\
      0 &  &  & -1 & 2 \\
\end{bmatrix*}\in \Rr^{n_x\times n_x},
\end{align}
and $\bv=(\Delta x)^2(f(x_1),\cdots,f(x_{n_x}))^T$. For simplicity, let us assume that $n_x= 2^s$. We can write a recursive decomposition for $\Am_e$ in Sigma basis as follows
\begin{align}\label{eq: sigma_x}
\Am_e \coloneqq \Am^{(s)} &= \Id_2 \otimes \Am^{(s-1)} + \Sm \otimes \Dm^{(s-1)} \nonumber\\
&\qquad + \Sp \otimes (\Dm^{(s-1)})^T, \nonumber\\
\Dm^{(s-1)} &= -\underbrace{\Sp \otimes \dots \otimes \Sp}_{s-1 \text{ times}}, \nonumber\\
\Am^{(1)} &= 2\,\Id_2 - \Sm - \Sp.
\end{align}
Thus, $\Am_e$ can be decomposed into $2s+1=2\log n_x+1$ terms, which logarithmically depend on the grid size $n_x$.~\Cref{fig: lcucomp}a illustrates the exponential reduction in the number of terms with the Sigma basis as compared to the Pauli basis as a function of the grid size $n_x$. 

As shown in \cite{liu2021variational}, a similar logarithmic dependence on $n_x$ holds for 1D Poisson equation with Neuman and Robin boundary conditions as well. The result also extends to $d$-dimensional Poisson PDE with Dirichlet boundary condition. 

\subsection{Parabolic PDE: Heat Equation}\label{sec: paraPDE}
\label{subsec: heat}
Consider the 1D heat equation over the domain $[0,\, l]$ with Neumann boundary conditions,
\begin{align}
\frac{\partial u}{\partial t}&=\alpha \frac{\partial^2 u}{\partial x^2},\nonumber\\
  -k\frac{\partial u}{\partial x}\bigg|_{x=0} = q, &\,
  -k\frac{\partial u}{\partial x}\bigg|_{x=l} = 0, \nonumber\\
   u(x,0)&=u_0(x),
\end{align}
where $\alpha$ is the thermal diffusivity, $k$ is the thermal conductivity of the material and $q$ is a constant heat flux. We discretize the PDE in space with second-order accuracy, and the boundary condition with first-order accuracy, leading to a system of first-order ODEs
\begin{equation}\label{eq: heatode}
\frac{d \uv(t)}{d t}=\frac{\alpha}{(\Delta x)^2}\Am_p \uv(t),
\end{equation}
where $\uv(t)=\begin{bmatrix}
u(x_1,t) & u(x_2,t) & \dots & u(x_{n_x},t)
\end{bmatrix}^T$ with $x_i=(i-1)\Delta x$, $i=1,\cdots,n_x$, $n_x$ is the number of spatial grid points, $\Delta x=l/n_x$ is the spatial grid size,  $\ev_1=\begin{bmatrix} 1 & 0 & 0 \cdots 0 \end{bmatrix}^T\in \mathbb{R}^{n_x}$, and $\Am_p$ is a $n_x \times n_x$ matrix of the form
\begin{align}\label{eq: heatA}
\Am_p &= \begin{bmatrix*}[r]
 -1 & 1 &  &  & 0 \\
       1 & -2 & 1 &  &  \\
       & \ddots & \ddots & \ddots &  \\
       &  & \ddots  & -2 & 1 \\
      0 &  &  & 1 & -1 \\
\end{bmatrix*}.
\end{align}
Backward Euler scheme is used for the time discretization of~\Cref{eq: heatode}, which leads to a system of difference equations of the form
\begin{align}
\label{eq: diff}
\bigg(\Id_{n_x}-\frac{\alpha \Delta t}{(\Delta x)^2} \Am_p\bigg)\uv_{k+1}= \uv_k+ \frac{q\,\Delta t}{k\Delta x}\ev_1,
\end{align}
where $\uv_k=\uv((k-1)\Delta t)$, $k=1,\cdots,n_t$ with $\Delta t=T/(n_t-1)$ being the temporal grid size. We can express the difference equations of the form~\Cref{eq: diff}  into a single linear system $\Am_h \uv = \bv$, where, $\uv = \begin{bmatrix}
\uv_1^T & \uv_2^T & \dots & \uv^T_{n_t}
\end{bmatrix}^T$, $\bv = \begin{bmatrix}
\uv_0^T & \frac{q\Delta t}{k \Delta x}\ev^T_1 & \dots & \frac{q\Delta t}{k \Delta x}\ev^T_1
\end{bmatrix}^T$ and 
\begin{align}\label{eq: full_coeff}
\Am_h &= \begin{bmatrix*}[r]
\Id_{n_x} & & & 0\\
-\Id_{n_x} & \Id_{n_x} & & \\
& \ddots & \ddots & \\
0 & & -\Id_{n_x} & \Id_{n_x}
\end{bmatrix*} \nonumber \\
&- \frac{\alpha \Delta t}{\Delta x^2} \begin{bmatrix}
0 & & & \\
 & \Am_p & & \\
 &&\ddots & \\
 && &\Am_p
\end{bmatrix} = \Am_1 - \frac{\alpha \Delta t}{\Delta x^2}  \Am_2.
\end{align}
For simplicity, let us assume that $n_x= 2^s$ and $n_t = 2^t$. Then $\Am_1$ can be written using $\log n_t + 1$ terms as shown in~\Cref{sec: methoddecomp}.
Next, $\Am_2$ can be written as
\begin{align*}
\Am_2 &= \Id_{n_t} \otimes \Am_p - \Sp\Sm \otimes \dots \otimes \Sp\Sm \otimes \Am_p,
\end{align*}
with
\begin{align}\label{eq: heat_bc}
\Am_p &= \begin{bmatrix*}[r]
 -2 & 1 &  &   0 \\
       1 & \ddots &\ddots &    \\
       & \ddots & \ddots & 1   \\
      0 &    & 1 & -2 \\
\end{bmatrix*}  + \begin{bmatrix*}[r]
 1 &  &  &    \\
        &0 & &    \\
       &  & \ddots &    \\
       &   &  & 1 \\
\end{bmatrix*} \nonumber \\
&=\Am_{p1} + \Am_{p2}. 
\end{align}
$\Am_{p2}$ is simply two terms $\Sp\Sm \otimes \dots \otimes \Sp\Sm + \Sm\Sp \otimes \dots \otimes \Sm\Sp$. 
Since $\Am_{p1}=-\Am_e$ where $\Am_e$ is defined in~\Cref{eq:possionA}, we can use the same recursive procedure to obtain the decomposition of $\Am_{p1}$. Thus, we can write $\Am_p$ using $2\, \log n_x +3$ terms and, hence, $\Am_2$ using $4 \,\log n_x + 6$ terms. Figure~\ref{fig: lcucomp}b compares the number of terms for the Pauli and Sigma basis decomposition, again showing the significant efficiency obtained by using Sigma basis.

We obtain a similar decomposition given a Robin boundary condition of the form $w_1 u(x,t) + w_2 \frac{\partial u}{\partial x} = q$ with the only difference being the non-zero entries of $\Am_{p2}$ taking the value $w_2 / (w_1 \Delta x + w_2)$. Note that, the Neumann and Dirichlet boundary conditions can be obtained by setting $w_1$ and $w_2$ to zero, respectively.

\subsection{Hyperbolic PDE: Wave Equation}
Consider the 1D wave equation over the domain $[0,l]$ with Neumann boundary conditions,
\begin{align*}
& \frac{\partial^2 w(x,t)}{\partial t^2}=c^2\frac{\partial ^2 w(x,t)}{\partial x^2}, \\
& \frac{\partial w(x,t)}{\partial x}\bigg|_{x=0}=0,\quad \frac{\partial w(x,t)}{\partial x}\bigg|_{x=l}=0,
\end{align*}
and initial conditions,
\begin{equation*}
w(x,0)=f_1(x), \quad \frac{\partial w(x,t)}{\partial t}\bigg|_{t=0}=f_2(x).
\end{equation*}
Following similar spatial discretization procedure as for the 1D heat equation, we get a system of second-order ODEs
\begin{equation}\label{eq: wave2ode}
\frac{d^2 \wv(t)}{d t^2}=\frac{c^2}{(\Delta x)^2}\Am_p \wv(t),
\end{equation}
where $\Am_p$ is as defined in~\Cref{eq: heatA}, $\wv(t) = \begin{bmatrix}
w(x_1,t) & w(x_2,t) & \dots & w(x_{n_x},t)
\end{bmatrix}^T$,  $x_i=(i-1)\Delta x, i=1,\cdots,n_x$ with $\Delta x=l/n_x$, and initial condition
\begin{equation}\label{eq: waveinit}
\wv(0)=\fv_1, \quad \frac{d \wv(t)}{d t}\bigg|_{t=0}=\fv_2,
\end{equation}
with $\fv_1=\begin{bmatrix}f_1(x_1),\cdots,f_1(x_{n_x})\end{bmatrix}^T$ and 
$\fv_2=\begin{bmatrix}f_2(x_1),\cdots,f_2(x_{n_x})\end{bmatrix}^T$. Introducing
\begin{equation*}
\vv(t)=\frac{d \wv(t)}{d t},\quad \uv(t)=\begin{bmatrix}
          \wv(t)  \\
          \vv(t)\\
      \end{bmatrix}\in \Rr^{2n_x},
\end{equation*}
we express~\Cref{eq: wave2ode} as a system of first order ODEs
\begin{equation}
\frac{d \uv(t)}{d t}=\tilde{\Am}_p \uv(t),
\end{equation}
with initial condition $\uv(0)=\begin{bmatrix}\fv_1^T,\fv_2^T\end{bmatrix}^T$  and 
\begin{align}\label{eq: waveA}
&\tilde{\Am}_p=\begin{bmatrix}
          0 & \Id_{n_x} \\
          \frac{c^2}{(\Delta x)^2}\Am_p & 0 \\
      \end{bmatrix}\in \Rr^{2n_x\times 2n_x}.
\end{align}
Finally, following the temporal discretization steps as in the previous section, we obtain a linear system $\Am_w\uv=\bv$ where
\begin{align}\label{eq: full_coeffwave}
\Am_w &= \begin{bmatrix*}[r]
\Id_{2n_x} & & & 0\\
-\Id_{2n_x} & \Id_{2n_x} & & \\
& \ddots & \ddots & \\
0 & & -\Id_{2n_x} & \Id_{2n_x}
\end{bmatrix*} \nonumber \\
&- \Delta t \begin{bmatrix}
0 & & & \\
 & \tilde{\Am}_p & & \\
 &&\ddots & \\
 && &\tilde{\Am}_p
\end{bmatrix},
\end{align}
and $\bv=\begin{bmatrix}(\uv(0))^T,0,\cdots,0\end{bmatrix}^T$.  Since
\begin{equation}\label{eq: Awavedecomp}
\tilde{\Am}_p=\Sm\otimes\frac{c^2}{(\Delta x)^2}\Am_p+\Sp\otimes \Id_{n_x},
\end{equation}
following similar procedure as for the 1D heat equation, $\tilde{\Am}_p$ can be expressed with $2\log 2n_x+4$ LCNU terms, and thus $\Am_w$ requires a total of $\log n_t+1+2(2\log 2n_x+4)$ terms. 
Figure~\ref{fig: lcucomp} compares the number of terms for the Pauli and Sigma basis decomposition with trends similar to those for the Poisson and heat equations.

\subsection{Higher Dimensional Systems}

The structure of the \( \Am_e \) and \( \Am_p \) matrices in~\Cref{eq:possionA,eq: heatA,eq: Awavedecomp} depends on the discretization scheme, spatial dimensions, and boundary conditions used. In higher dimensions, assuming the same finite difference scheme, the structured pattern of nonzero entries can still be efficiently captured using the Sigma basis, depending on the choice of boundary conditions. For example, Dirichlet boundary conditions do not affect the nonzero structure of \( \Am_e \) in~\Cref{eq:possionA}, as they only modify the right-hand-side vector \( \bv \). As a result, the same polylogarithmic scaling observed in the 1D case holds.

For Neumann and Robin boundary conditions, however, the boundary terms must be treated separately, as in \( \Am_{p2} \) from~\Cref{eq: heat_bc}. The number of such terms, denoted \( n_b \), scales as \( \Theta(n_x^{\frac{d-1}{d}}) \), where \( n_x \) is the \textit{total} number of grid points in \( d \) dimensions. If these boundary terms are position-dependent, up to \( n_b \) LCNU terms may be needed to represent \( \Am_{p2} \), which diminishes the efficiency gained from using the Sigma basis. However, if the boundary terms exhibit sufficient spatial uniformity, \( \Am_{p2} \) retains a structured sparsity pattern. This allows its LCNU decomposition to be performed using the semi-analytical approach described in~\Cref{sec: methoddecomp}, potentially recovering polylogarithmic scaling. This must be assessed on a case-by-case basis, as there is no general recipe.

\subsection{Carleman Linearized Non-linear PDEs}

Nonlinear PDEs cannot be directly simulated on a quantum computer. For PDEs with polynomial nonlinearities, \emph{Carleman linearization} (CL) has been proposed as a method to transform the nonlinear system into an infinite set of linear ODEs, which are subsequently truncated into a finite system of linear ODEs \cite{liu2021efficient,li2025potential,demirdjian2022variational,surana2024efficient}. These linear ODEs can be discretized in time to yield a linear system, similar to the examples discussed earlier. The CL transformation produces a specific structure in the resulting linear system that can be exploited using the Sigma basis decomposition. For the 1D nonlinear Burgers’ equation, it was shown in~\cite{gnanasekaran2024variational} that the number of LCNU terms under the Sigma basis scales polylogarithmically with the temporal grid size, offering a significant advantage over the Pauli basis. This result was further improved in~\cite{demirdjian2025efficient}, where polylogarithmic scaling was established with respect to both spatial and temporal grid sizes. We refer readers to these works for a detailed treatment of Carleman-linearized nonlinear PDEs.

\section{Utilizing sigma decomposition in variational and fault-tolerant quantum algorithms}
\label{sec: uses_lcu}

LCU serves as a foundational technique for algorithms in both the fault-tolerant quantum computing (FTQC) and noisy intermediate-scale quantum (NISQ) eras. In the NISQ setting, LCU plays a key role in VQAs, particularly in the design and efficient evaluation of cost functions. In the FTQC regime, LCU commonly appears through block-encoded representations of matrices. In this section, we show how our LCNU decomposition using the Sigma basis can be applied in both contexts --- supporting VQAs and facilitating block-encoding of structured matrices.

\subsection{LCU for VQAs}
In VQAs, one typically has to compute terms of the form $\bra{\psi_1} \Am_l \ket{\psi_2}$ and $\bra{\psi_1} \Am_i^* \,\mathbf{M}\, \Am_j \ket{\psi_2}$ for arbitrary states $\ket{\psi_1}$, $\ket{\psi_2}$ and unitary matrix $\Mm$. Evaluations of these terms are typically combined to compute the desired cost function on $\Am$.  Let us assume that we are given unitary matrices $\Um$ and $\Vm$ to prepare states $\ket{\psi_1}$ and $\ket{\psi_2}$ as follows,
\begin{align*}
\ket{\psi_1}=\Um\ket{0},
\ket{\psi_2}=\Vm\ket{0}.
\end{align*}
In this section, we provide the Hadamard test circuits to compute terms of the form indicated above for $\Am_l=\otimes_k \sigma_k$, where $\sigma_k \in \mathbb{S}$. 

\subsubsection{Evaluation of $\bra{\psi_1} \Am_l \ket{\psi_2}$}

The Hadamard test circuit for computing $\bra{\psi_1} \Am_l \ket{\psi_2}$ is shown in~\Cref{fig: global_cost_circuit}.  Starting with a $n+2$ qubit system (2 ancilla qubits), the circuit performs the following sequence of operations:
\begin{align*}
\ket{0^{n+2}}& \xrightarrow[]{H_{a_0}, OC_\Um, C_{\Vm}} \frac{1}{\sqrt{2}} \big( \ket{00} \ket{\psi_1} +  \ket{10} \ket{\psi_2} \big)  \\
\xrightarrow[]{C_{\Um_l}}& \frac{1}{\sqrt{2}} \big( \ket{00} \ket{\psi_1} +  \ket{1}\Um_l\ket{0}\ket{\psi_2} \big)  \\
=&\frac{1}{\sqrt{2}} \big( \ket{00} \ket{\psi_1} + 
 \ket{10} \Am_l \ket{\psi_2}+ \ket{11} \Am_l^c \ket{\psi_2}  \big)  \\
\xrightarrow[]{H_{a_0}} &\frac{1}{2} \Big(\ket{00} \big(\ket{\psi_1}+\Am_l \ket{\psi_2}\big)\\
&\quad+\ket{10} \big(\ket{\psi_1} - \Am_l \ket{\psi_2}\big) \\
&\quad+\ket{01} \Am_l^c \ket{\psi_2}-\ket{11} \Am_l^c \ket{\psi_2}\Big).
\end{align*}
Here $C_*, OC_*$ represent controlled and open-controlled application of the unitaries.  After measuring the two ancilla qubits and using the relation,
\begin{align*}
&P_{00} - P_{10} \\
&= \frac{1}{4} \left( \left\| \ket{\psi_1} + \Am_l \ket{\psi_2} \right\|^2 
                     - \left\| \ket{\psi_1} - \Am_l \ket{\psi_2} \right\|^2 \right) \\
&= \operatorname{Re} \bra{\psi_1} \Am_l \ket{\psi_2} \\
&= \operatorname{Re} \bra{0} \Um^* \Am_l \Vm \ket{0}
\end{align*}
gives the real part of the desired desired result. The imaginary part can be computed by inserting the phase gate on $a_0$. The unitary operator $\Um_l$, as shown in the blue boxes in \Cref{fig: global_cost_circuit}, can be implemented efficiently as discussed in~\Cref{subsec:Um}.

\begin{figure*}[htbp]
\centering
\begin{subfigure}[t]{0.42\textwidth}
\resizebox{1.\textwidth}{!}{
\begin{quantikz}[wire types = {q, q, q, q, n, q, q}, transparent]
\lstick{$a_0$ \quad \ket{0}}& \gate{H} & \ctrl{2} & \octrl{2}   &\ctrl{1} & \ctrl{1} &  \gate{H} & \meter{}  \\
\lstick{$a_1$ \quad \ket{0}}&   &&       &        \gate{X} \wire[d][1]{q}\gategroup[6,steps=2,style={dashed,rounded
corners,fill=blue!20, inner
xsep=2pt},background,  label style={label
position=below,anchor=north,yshift=-0.2cm}]{{\sc ${\Um}_l$}}  & \targ{} && \meter{}  \\
\lstick{$q_0$ \quad \ket{0}}& & \gate[5]{\Vm} & \gate[5]{\Um} &  \gate[5]{\bar{\Am}_l} &  \octrl{-1}  &  & \\
\lstick{$q_1$ \quad \ket{0}} &&  && &\ctrl{-1}  & & \\
\vdots & && && \vdots&   &\\
\lstick{$q_{n-2}$ \quad \ket{0}}& && & &\octrl{1}  && \\
\lstick{$q_{n-1}$ \quad \ket{0}}& && &  &\ctrl{-1} & & 
\end{quantikz}
}
\caption{Circuit for computing $\bra{0}\Um^* \Am_l \Vm\ket{0}$.}
\label{fig: global_cost_circuit}
\end{subfigure}
~
\begin{subfigure}[t]{0.55\textwidth}
\resizebox{1.\textwidth}{!}{
\begin{quantikz}[wire types = {q, q, q, q, n, q, q}, transparent]
\lstick{$a_0$ \quad \ket{0}}& \gate{H} & \ctrl{2} & \octrl{2} & \ctrl{1} & \ctrl{1}   & \ctrl{1} &   \octrl{1} & \octrl{1} & \gate{H} & \meter{}  \\
\lstick{$a_1$ \quad \ket{0}}&  & &         &   \gate{X}\wire[d][1]{q} \gategroup[6,steps=2,style={dashed,rounded
corners,fill=blue!20, inner
xsep=2pt},background, label style={label
position=below,anchor=north,yshift=-0.2cm}]{{\sc
${\Um}_j$}}    & \targ{}  & \octrl{1} &  \gate{X}\wire[d][1]{q}\gategroup[6,steps=2,style={dashed,rounded
corners,fill=blue!20, inner
xsep=2pt},background, label style={label
position=below,anchor=north,yshift=-0.2cm}]{{\sc
${\Um}_i$}} & \targ{} & & \meter{}  \\
\lstick{$q_0$ \quad \ket{0}}& & \gate[5]{\Vm} & \gate[5]{\Um} &  \gate[5]{\bar{\Am}_j}&  \octrl{-1}   & \gate[5]{\mathbf{M}}   & \gate[5]{\bar{\Am}_i} & \ctrl{-1}& & \\
\lstick{$q_1$ \quad \ket{0}} & && & & \ctrl{-1}   && & \octrl{-1}& &\\
\vdots & && && \vdots  &&  &  \vdots &  & \\
\lstick{$q_{n-2}$ \quad \ket{0}}& && & & \octrl{0} && & \ctrl{0} & &\\
\lstick{$q_{n-1}$ \quad \ket{0}}& && & & \ctrl{-1} && & \ctrl{-1} & &
\end{quantikz}
}
\caption{Circuit for computing $\bra{0}\Um^* \Am_i^* \mathbf{M} \Am_j \Vm\ket{0}$. }
\label{fig: local_cost_circuit}
\end{subfigure}
\caption{Hadamard test circuits for computing two terms which typically arise in VQAs. }
\end{figure*}
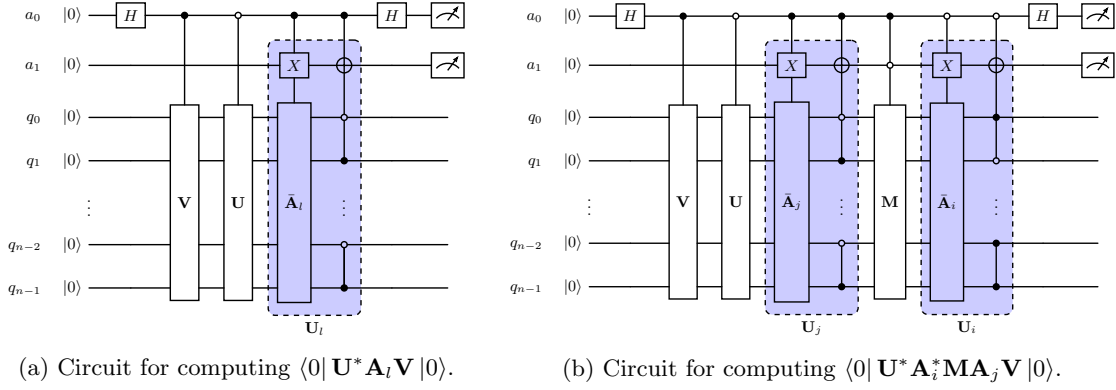

\subsubsection{Evaluation of $\bra{\psi_1} \Am_i^* \,\mathbf{M}\, \Am_j \ket{\psi_2}$}
To compute $\bra{\psi_1} \Am_i^* \mathbf{M} \Am_j \ket{\psi_2}$, we start similarly with a $n+2$ qubit system and apply the following operations as shown in~\Cref{fig: local_cost_circuit}:
\begin{align*}
\ket{0^{n+2}} &\xrightarrow[]{H_{a_0}, OC_\Um, C_{\Vm}} \frac{1}{\sqrt{2}} \big( \ket{00} \ket{\psi_1} +  \ket{10} \ket{\psi_2} \big)   \\
\xrightarrow[]{OC_{\Um_i}, C_{\Um_j}}& \frac{1}{\sqrt{2}} \big( \ket{0} {\Um}_i \ket{0}\ket{\psi_1} +  \ket{1} {\Um}_j \ket{0}\ket{\psi_2} \big)  \\
&\hspace{-15pt}= \frac{1}{\sqrt{2}} \big( \ket{00} \Am_i \ket{\psi_1} + \ket{01} \Am_i^c \ket{\psi_1} +\\
&\qquad   \ket{10} \Am_j \ket{\psi_2} +  \ket{11} \Am_j^c \ket{\psi_2} \big)  \\
\xrightarrow[]{COC_{\mathbf{M}}} &\frac{1}{\sqrt{2}} \big( \ket{00} \Am_i \ket{\psi_1} +  \ket{01} \Am_i^c \ket{\psi_1}+ \\
&\qquad \ket{10}\mathbf{M} \Am_j \ket{\psi_2} +  \ket{11} \Am_j^c \ket{\psi_2} \big)  \\
\xrightarrow[]{H_{a_0}} &\frac{1}{2} \bigg( \ket{00} \big(\Am_i \ket{\psi_1} + \Mm \Am_j \ket{\psi_2} \big) + \\
&\qquad\ket{10} \big( \Am_i \ket{\psi_1} - \Mm \Am_j \ket{\psi_2}  \big) + \\
  & \qquad \ket{01}\big(\Am_i^c\ket{\psi_1}+\Am_j^c\ket{\psi_2}\big)+\\
  &\qquad\ket{11}\big(\Am_i^c\ket{\psi_1}-\Am_j^c\ket{\psi_2}\big)  \bigg),
\end{align*}
where $ COC_*$ represent control-open control application of the unitary. Finally, measuring the two ancilla qubits 
\begin{align*}
P_{00} - P_{10} &= \frac{1}{4} \bigg( \Big\| \Am_i \ket{\psi_1} + \Mm \Am_j \ket{\psi_2}  \Big\| ^2 - \\
&\qquad\quad\Big\|  \Am_i \ket{\psi_1} - \Mm \Am_j \ket{\psi_2} \Big\| ^2\bigg) \\
&= \operatorname{Re}\bra{\psi_1} \Am_i^* \mathbf{M} \Am_j \ket{\psi_2},
\end{align*}
leads to the desired result. The imaginary part can be calculated similarly as indicated above.

\begin{remark}
    We have described the construction of Hadamard test circuits to compute the desired inner products with the LCNU terms. One can also construct Hadamard Overlap test circuits in a similar fashion to avoid the controlled application of the unitaries $\Um, \Vm$. 
\end{remark}

\subsubsection{Resource Estimation}

As highlighted throughout this paper, our LCNU technique using the Sigma basis can yield an exponential reduction in the number of terms compared to standard LCU decomposition. However, the quantum circuit for \( \Um_l \) involves a multi-controlled \( C^mX \) gate, which increases the depth of Hadamard test circuits. In this section, we show that the circuits can be implemented with a minimal increase in the gate size and circuit depth.

Efficient implementation of \( C^mX \) gates is an active area of research. These gates are typically decomposed into elementary universal gate sets such as CNOT and single-qubit gates. Any such decomposition must have a depth of \( \Omega(\log m) \) and a size of \( \Omega(m) \)~\cite{10.5555/2011679.2011682}. A recent construction~\cite{nie2024quantum} implements the \( C^mX \) gate using only single-qubit and CNOT gates, achieving a depth of at most \( \mathcal{O}(\log m) \) and size \( \mathcal{O}(m) \), with one ancilla qubit. In our setting, \( m \leq \log N \), where \( N \) is the matrix size, so the additional overhead remains logarithmic in problem size. This modest increase in circuit complexity is far outweighed by the exponential reduction in the number of LCU terms achieved through the Sigma basis. Moreover, emerging hardware platforms such as trapped ions, Rydberg atoms, and superconducting circuits may support direct implementation of \( C^mX \) gates in the future~\cite{katz2022n}.

\subsection{LCU for FTQC algorithms}

Block-encoding is a central data access model for representing matrices in fault-tolerant quantum algorithms. A common approach for constructing efficient block-encodings is via the LCU decomposition of the target matrix~\cite{lin2022lecturenotesquantumalgorithms}. Although our LCNU decomposition uses non-unitary operators from the Sigma basis, we show that it can still be used to efficiently construct block-encodings. In the following section, we describe this procedure and compare its resource cost with the standard Pauli-basis LCU approach. We note that other techniques have also been developed to exploit matrix sparsity and structure in block-encoding constructions, including~\cite{camps2022fable,camps2024explicit,sunderhauf2024block}.
 
Given the standard LCU decomposition of a matrix $\Am$ with $L$ terms, it can be block-encoded using the PREP and SELECT operators defined below. Without loss of generality, assume that $\log_2 L$ is an integer and $\sum_{l=0}^{L} \alpha_l = 1$. 
\begin{align*}  
\ket{\mathbf{\alpha}} \coloneqq \operatorname{PREP} \ket{0} &= \sum_{l=1}^{L}\sqrt{\alpha_l}\ket{l}, \\
\operatorname{SELECT}\ket{l}\ket{\psi} &= \ket{l}\Am_l\ket{\psi}.
\end{align*}

Despite the non-unitarity of the operators $\Am_l$ in the LCNU decomposition under the Sigma basis, we can efficiently construct a block-encoding of $\Am$. Block-encodings have extensibility properties, i.e., given block-encodings of matrices $\mathbf{P}$ and $\mathbf{Q}$, we can get the block-encoding of $c_0\mathbf{P} + c_1 \mathbf{Q}$ under mild conditions on $c_0, c_1$~\cite{tang_qsvt_lecture}.  As each \( \Am_l \) is block-encoded via unitary completion to yield unitary \( \Um_l \), we can apply this linear combination of block-encodings approach to naturally encode $\Am$. 

\subsubsection{Resource Estimation}

In order to compare the resource estimates for Pauli and Sigma basis block-encoding, we use the protocols for the PREP and SELECT operators (for Pauli strings) described in~\cite{zhang2024circuit} that were shown to have near-optimal gate complexities. The number of ancilla qubits, gate count, and circuit depth for these routines are summarized in~\Cref{table:comp_sp_so}. 

\begin{table*}[t]
\centering
\begin{tabular}{cccc}
\toprule
Protocol & $n_{\text{anc}}$ & Count & Depth  \\
\midrule
PREP  & $\Omega(\log{L})\leq n_\text{anc} \leq \mathcal{O}(L)$  & $\mathcal{O}(L \log{1/\epsilon}) $ & $\Tilde{\mathcal{O}}(L \log (1/\epsilon)\frac{\log{n_\text{anc}}}{n_{\text{anc}}})$\\ \\
SELECT  & $\Omega(\log{L} + \log N) \leq n_\text{anc} \leq \mathcal{O}(L\log{N})$  & $\mathcal{O}(L\log{N})$  & $\mathcal{O}(L\log{N}\frac{\log{n_\text{anc}}}{n_\text{anc}})$  \\
\bottomrule
\end{tabular}
\caption{Clifford$+$T complexity of state preparation of $\ket{\alpha}$ of size $L$ and $\text{Select}(A_l)$ for Pauli strings, where $\epsilon$ is the desired accuracy in preparing the state. The $\Tilde{\mathcal{O}}$ hides doubly logarithmic factors.}
\label{table:comp_sp_so}
\end{table*}

Let us estimate the circuit complexity for the SELECT operator for $\Um_l$ by following a similar procedure as in~\cite{zhang2024circuit}. First, the SELECT operator for $\Um_l$ is defined as follows: 
\[
\operatorname{SELECT}\ket{l}\ket{0}\ket{\psi} = \ket{l} \Um_l\ket{0}\ket{\psi},
\]
while the $\operatorname{PREP}$ operator remains the same as defined above. Note that one ancilla bit is needed to block-encode $\Am_l$ in the unitary $\Um_l$.

We apply Lem. 6, 7 and the proof of Thm. 4 in~\cite{zhang2024circuit}. Redefine, $C_{\text{ctrl}}(\Um_l, r)$ and $D_{\text{ctrl}}(\Um_l, r)$ as the count and depth of Clifford $+$ T gates required to construct a single-qubit controlled-$\Um_l$ given $r$ ancillary qubits. The circuit for $\Um_l$ consists of single qubit gates and a $C^m X$ gate  where $m \leq n$ and is implemented on $n + 1$ qubits as we discussed in~\Cref{thm: circuit_U}. Given, $n+1$ ancillary qubits, controlled-$\Um_l$ can be constructed with the circuit shown in~\Cref{fig: blenc_Ul}.
\begin{figure}
    \centering
    \begin{quantikz}
\lstick{$c$}&\ctrl{7}& \ctrl{1}&\ctrl{8}&\ctrl{7}&\qw\\
\lstick{$a$}&\qw& \gate{X} \wire[d][0]{q}\gategroup[8,steps=2,style={dashed,rounded
corners,fill=blue!20, inner
xsep=2pt},background,  label style={label
position=below,anchor=north,yshift=-0.2cm}]{} & \targ{} &\qw&\qw\\
\lstick{$a_0$}&\targ{}&\ctrl{1}& \qw&\targ{}&\qw\\
\lstick{$q_{0}$}&\qw&\gate{\bar{\Am}_{l,0}}&\octrl{1}&\qw&\qw\\
\lstick{$a_1$}&\targ{}&\ctrl{1}&\qw&\targ{}&\qw\\
\lstick{$q_{1}$}&\qw&\gate{\bar{\Am}_{l,1}}&\ctrl{1}\qw&\qw & \qw\\
\cdots&& \cdots&&&\cdots\\
\lstick{$a_{n-1}$}&\targ{}&\ctrl{1}&\qw&\targ{}&\qw\\
\lstick{$q_{n-1}$}&\qw&\gate{ \bar{\Am}_{l,n-1}}&\ctrl{0}&\qw&\qw 
\end{quantikz}
    \caption{Circuit for constructing controlled-$\Um_l$ with $n+1$ ancilla bits. The $\bar{\Am}_{l,*}$ are the tensor product terms of the corresponding completion operator $\bar{\Am}_l$.}
    \label{fig: blenc_Ul}
\end{figure}
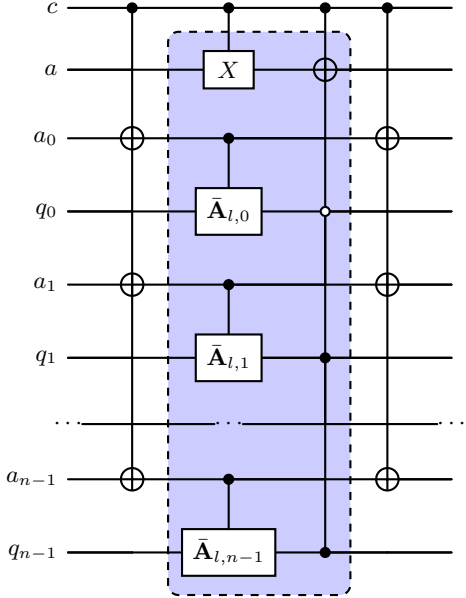

The two $n-$Toffoli gates (one control, $n$ targets) can be constructed effectively with $\mathcal{O}(n)$ count and $\mathcal{O}(\log n)$ depth of Clifford$+$T gates. The $C^m X$ gate can be constructed with the same asymptotic gate count and depth with an additional ancilla qubit\cite{nie2024quantum}. Thus $C_{\text{ctrl}}(\Um_l, n+2) = \mathcal{O}(n)$, $D_{\text{ctrl}}(\Um_l, n+2) = \mathcal{O}(\log n)$.  

Finally using the procotols defined in Lem 6, 7 and following the same analysis in the proof of Thm. 4 in~\cite{zhang2024circuit}, we obtain the same asymptotic gate count and circuit depth for Sigma basis decomposition that is summarized in~\Cref{table:comp_sp_so} at the expense of two additional ancilla bits. 

Thus, the comparison for gate complexities for block-encoding based on Pauli basis and Sigma basis depend purely on the number of terms $L$ in the decomposition. Under the assumption that $L$ scales as $\mathcal{O}(N)$ with Pauli and as $\mathcal{O}(\text{poly}\log N)$ with Sigma basis, we obtain an exponential improvement in gate count, circuit depth, and ancilla qubits with our Sigma basis LCNU approach. 

\section{Unitary Completion vs Unitary Dilation}\label{sec: discussion}

One might wonder how our approach based on unitary completion compares to the unitary dilation technique. We will compare and contrast the two related techniques in this section. 

The unitary dilation of a matrix is given by~\Cref{def:dilation}. Note that the dilation operator, $\Dm_{\Am}$, is well defined as $\In-\Am^{*}\Am\geq 0$ is a positive semi-definite matrix under the contraction assumption. 

\begin{definition}
Let $\Am$ be contraction, i.e., $\|\Am\|_2\leq 1$, then the unitary dilation $\Um_{\Am}$ of $\Am$ is defined as:
\begin{equation}
\Tilde{\Um}_{\Am}=\left(
            \begin{array}{cc}
              \Am & \Dm_{\Am^{*}} \\
              \Dm_{\Am} & -\Am^{*} \\
            \end{array}
          \right), \label{eq:ud}
\end{equation}
where, $\Dm_{\Am}=\sqrt{\In-\Am^{*}\Am}$. 
\label{def:dilation}
\end{definition}

As per the Sz.-Nagy dilation theorem, every contraction on a Hilbert space has a unitary dilation which is unique up to an unitary equivalence \cite{schaffer1955unitary}. In contrast, the concept of unitary completion only makes sense when the non-zero columns of $\Am$ are unitary according to~\Cref{def: completion}. Thus, unitary completion is not applicable to every contraction on the Hilbert space.  

Both dilation and completion exist for the Sigma basis LCNU terms, $\Am_l$. Let us focus on understanding this relation for the remainder of the section. First, note that the dilation of $\Am_l$, $\tilde{\Um}_l \coloneqq \tilde{\Um}_{\Am_l}$, on an ancillary system of the form $\ket{0} \ket{\psi}$ gives
\begin{equation*}
\tilde{\Um}_{l}|0\ra|\psi\ra=|0\ra \Am_l  |\psi\ra+|1\ra \Dm_{\Am_l} |\psi\ra,
\end{equation*}
similar to the expression in \Cref{eq: Umact}.
The term $(\In-\Am^*_l\Am_l)^2$ can be simplified as $\In-\Am^T_l\Am_l$ using properties of the Sigma basis given in~\cite{gnanasekaran2024efficient}. Thus, we can simplify $\tilde{\Um}_l$ as
\begin{equation}\label{eq: Um_dil}
\tilde{\Um}_{l}=\begin{bmatrix}
	\Am_l &\In-\Am_l\Am_l^T	\\
              \In-\Am^{T}_l\Am_l  & -\Am^{T}_l  \\
            \end{bmatrix}.
\end{equation}
Note that $\tilde{\Um}_{l}$ is unitary even without the negative sign in the (2,2) block position and can be shown using the properties of the Sigma basis. To simplify analysis, let us ignore the negative sign as it is unitarily equivalent (up to a controlled $\sigma_z$ gate). With this, $\tilde{\Um}_{l}$ is also a permutation matrix as we show below and can be expressed as a sequence of Toffoli gates. 
 
~\Cref{thm: Udil} shows that in general $\Tilde{\Um}_l$ requires more $C^n X$ gates than $\Um_l$ and hence is less efficient compared to our unitary completion approach.  Note that while there is a reduction in the number of single qubit gates, the bottleneck in circuit complexity arises from requiring \textit{multiple} $C^n X$ gate. With our completion-based approach we \textit{always} require \textit{only} a single $C^n X$. 

\begin{theorem}\label{thm: Udil}
    $\Tilde{\Um}_l$ as defined in~\Cref{eq: Um_dil} can be implemented using a single qubit gate and $2s+1$ $C^m X$ gates where $s = \sum_{i} \left[ \sigma_i \in \{ \Sp, \Sm \} \right]$ and $m \leq n$.
\end{theorem}
\begin{proof}
We can write $\Tilde{\Um}_l$ as
\begin{align*}
    \Tilde{\Um}_l &= \Tilde{\Um}_{l,1}\Tilde{\Um}_{l,2} \\
    &= \begin{bmatrix}
	\In-\Am_l\Am_l^T & \Am_l	\\
             \Am^{T}_l & \In-\Am^{T}_l\Am_l    \\
            \end{bmatrix}  \begin{bmatrix}
                 & \Id \\
                 \Id & 
            \end{bmatrix}.
\end{align*}
We can construct $\Tilde{\Um}_{l,2}$ simply by applying $X$ gate on the ancilla qubit and $\Id$ on the other qubits. Let's focus on constructing $\Tilde{\Um}_{l,1}$.

Let $\Am_l = \sigma_0 \otimes \sigma_1 \otimes \dots \otimes \sigma_{n-1}$. Assume, that for certain indices $\mathbb{I} = \{i_1, \dots, i_s\}$, non-overlapping subsets $S_1 = \{ \Sp, \Sm\}$ and $S_2 = \{ \Sp\Sm, \Sm\Sp,\Id \}$, we have
\begin{align*}
   \sigma_{i \in \mathbb{I}} \in S_1, \quad 
   \sigma_{i \notin \mathbb{I}} \in S_2. 
\end{align*}
 $\Am_l$ is a binary matrix by construction. Suppose $\Am_l(r,c) = 1$. Using~\Cref{lem: Arow},
\[
\Am_l(r,c)=1 \iff \sigma_p(q_r(p), q_c(p)) = 1\quad \forall p,
\]
where, $r = \sum_{p=0}^{n-1} 2^{n-1-p}q_r(p)$ and $c = \sum_{p=0}^{n-1} 2^{n-1-p}q_c(p)$ are binary representations. By definition
\begin{align*}
    \sigma_i \in S_1 &\implies q_r(i) \neq q_c(i), \\
    \sigma_i \in S_2 &\implies q_r(i) = q_c(i). 
\end{align*} 
As the cardinality of the set $\mathbb{I}$ is $s$, the binary representation of $r$ and $c$ differ in $s$ bits. 

We have
\begin{align*}
\Am_l(r,c) = 1 &\implies
    \Am_l^T(c,r)=1, \\
    &\implies \tilde{\Um}_{l,1}(r,2^n+c) =1 \\
    &\implies \tilde{\Um}_{l,1}(2^n+c,r) =1. 
\end{align*}
Similarly, we also have
\begin{align*}
\Am_l(r,c) = 1 &\implies
    \Am_l\Am_l^T(r,r)=1, \\
    &\implies \tilde{\Um}_{l,1}(r,r) = 0 \\
    &\implies \tilde{\Um}_{l,1}(2^n+c,2^n+c) =0. 
\end{align*}
With these relations, we observe that $\Um_l$ is a permutation matrix. And, we can construct $\Tilde{\Um}_{l,1}$ by permuting rows $r$ and $2^n+c$ of the Identity matrix. As the binary representation of the two rows differ by $s+1$ bits, we require $2s+1$ $C^n X$ gates to permute them using~\Cref{cor:kbits_swap}. The exact set of control and target operations for each of these $C^n X$ gates can be inferred based on the proof of~\Cref{cor:kbits_swap} and~\Cref{rem:control_op}. 

If $k$ terms in the tensor product are $\Id$, then there are $2^k$ non-zero rows in $\Am_l$. Naively, one would then expect $2^k \, (2s + 1)$ $C^n X$ gates to construct $\Tilde{\Um}_{l,1}$. This would be true if we implement each set of $2s+1$ gates one after the other. However, similar to the proof of~\Cref{thm: Acomp}, we can optimize the circuit implementation to simply needing $2s+1$ $C^{n-k} X$ gates. 
\end{proof}

\begin{remark}
Note that when $s=0$ in~\Cref{thm: Udil},  $\tilde{\Um}_l = \Um_l$, i.e., the unitary produced by dilation is essentially the same (ignoring the phase factor) as the one produced by our completion technique. 
\end{remark}

\section{Linear Combination of ``Things"}

The Sigma basis set $\mathbb{S}$ defined in this work is an universal basis and can be used to compute a LCNU type decomposition for any matrix in $\mathbb{C}^{2^n \times 2^n}$. Trivially,  the Pauli matrices lie in the span of the Sigma basis.

Out of the five basis matrices in $\mathbb{S}$, only four are linearly independent, i.e. $\mathbb{S}$ in an over-complete basis. For example, $\Id = \Sp\Sm + \Sm\Sp$ and can be ignored. However, often times it is advantageous to keep $\Id$ in the basis set as it can lead to fewer terms in the decomposition as illustrated through various examples in~\Cref{sec: apps}. Similarly, any Pauli matrix can be added to the set as relevant to the problem. For example, in~\Cref{eq: sigma_x}, we can replace $\Sm + \Sp = \sigma_x$ and add Pauli-X matrix to $\mathbb{S}$. The corresponding completion operator for any Pauli matrix is itself and there is no control operation on the corresponding qubit (similar to how we handled $\Id$ in the proofs). Thus, we always have the flexibility to freely mix Pauli and Sigma basis terms to construct efficient LCNU type decompositions for a given problem. 

The Sigma basis was further generalized in \cite{demirdjian2025efficient} with the addition of certain permutation matrices and was used to efficiently represent Carleman linearized Burgers' equation. In general, we note that any matrix that is ``easy" to implement, i.e., simple unitaries with low-depth circuits, can be added to the Sigma basis. This flexibility makes our approach truly a linear combination on ``things", expanding the scope of LCU beyond conventional bases.  

\section{Conclusions}\label{sec: conc}

We introduced a novel framework for efficient LCNU decomposition of structured sparse matrices using the Sigma basis, achieving a polylogarithmic scaling in the number of decomposition terms with respect to matrix size. Unlike traditional Pauli-based approaches, our method leverages a simple set of \textit{non-unitary }operators. We addressed the non-unitarity via unitary completion—providing a simple and efficient quantum circuit construction with only marginal overhead compared to Pauli-based circuits.

We demonstrated how these unitary completion circuits can be seamlessly integrated into Hadamard-like test circuits for evaluating observables in VQAs. Furthermore, we extended this framework to construct block encodings of arbitrary operators given their Sigma basis decomposition, enabling their use in fault-tolerant quantum algorithms. We analytically established the resource requirements for Sigma basis LCNU based block encoding, showing an exponential reduction in the number of terms and overall circuit complexity compared to the Pauli basis LCU approach.

We illustrated our approach decomposition of matrices arising from several PDE discretizations and showed an exponential reduction in the number of decomposition terms compared to the Pauli basis. To support broader applicability, we also developed both numerical and semi-analytical tools for computing Sigma basis decompositions for arbitrary matrices. The Sigma basis LCNU technique has also been successfully applied in the study of nonlinear PDEs~\cite{demirdjian2025efficient} and in linear and nonlinear PDE-constrained optimization problems~\cite{surana2024variational,gnanasekaran2024variational}.

Looking ahead, the Sigma basis opens new avenues for efficient operator representations beyond quantum linear algebra applications. Future work may explore domain-specific Sigma-type bases, designed to exploit the intrinsic structure of problems across quantum simulation, machine learning, and optimization—extending the philosophy of linear combination of ``things” as a unifying abstraction for quantum algorithm design.

\section{Acknowledgments}
This research was developed with funding from the Defense Advanced Research Projects Agency (DARPA). The views, opinions, and/or findings expressed are those of the author(s) and should not be interpreted as representing the official views or policies of the Department of Defense or the U.S. Government.

\appendix
\section{CCNOT gate}\label{app: ccnot}

The truth table for the CCNOT gate with different control and target operations is given~\Cref{tab:truth_table}. This can be used to understand the placement of open and closed control operations on the qubits as given in~\Cref{rem:control_op}.

\begin{table*}[htbp]
\centering
\begin{tabular}{c c || c c}
\textbf{Circuit} & \textbf{Truth Table} & \textbf{Circuit} & \textbf{Truth Table} \\ 
\midrule
\begin{quantikz}[row sep=small, column sep=small]
\lstick{\(q_0\)} & \ctrl{2} & \qw \\
\lstick{\(q_1\)} & \ctrl{1} & \qw \\
\lstick{\(q_2\)} & \targ{}  & \qw
\end{quantikz}
&
\(
\begin{array}{ccc|ccc}
q_0 & q_1 & q_2 & q_0' & q_1' & q_2' \\ \hline
0 & 0 & 0 & 0 & 0 & 0 \\
0 & 0 & 1 & 0 & 0 & 1 \\
0 & 1 & 0 & 0 & 1 & 0 \\
0 & 1 & 1 & 0 & 1 & 1 \\
1 & 0 & 0 & 1 & 0 & 0 \\
1 & 0 & 1 & 1 & 0 & 1 \\
1 & 1 & 0 & \ab{1} & \ab{1} & \ab{1} \\
1 & 1 & 1 & \ab{1} & \ab{1} & \ab{0} \\
\end{array}
\) & \begin{quantikz}[row sep=small, column sep=small]
\lstick{\(q_0\)} & \octrl{2} & \qw \\
\lstick{\(q_1\)} & \octrl{1} & \qw \\
\lstick{\(q_2\)} & \targ{}  & \qw
\end{quantikz}
&
\(
\begin{array}{ccc|ccc}
q_0 & q_1 & q_2 & q_0' & q_1' & q_2' \\ \hline
0 & 0 & 0 & \ab{0} & \ab{0} & \ab{1} \\
0 & 0 & 1 & \ab{0} & \ab{0} & \ab{0} \\
0 & 1 & 0 & 0 & 1 & 0 \\
0 & 1 & 1 & 0 & 1 & 1 \\
1 & 0 & 0 & 1 & 0 & 0 \\
1 & 0 & 1 & 1 & 0 & 1 \\
1 & 1 & 0 & 1 & 1 & 0 \\
1 & 1 & 0  & 1 & 1 & 0  \\
\end{array}
\)\\
\midrule
\begin{quantikz}[row sep=small, column sep=small]
\lstick{\(q_0\)} & \ctrl{2} & \qw \\
\lstick{\(q_1\)} & \octrl{1} & \qw \\
\lstick{\(q_2\)} & \targ{}  & \qw
\end{quantikz}
&
\(
\begin{array}{ccc|ccc}
q_0 & q_1 & q_2 & q_0' & q_1' & q_2' \\ \hline
0 & 0 & 0 & 0 & 0 & 0 \\
0 & 0 & 1 & 0 & 0 & 1 \\
0 & 1 & 0 & 0 & 1 & 0 \\
0 & 1 & 1 & 0 & 1 & 1 \\
1 & 0 & 0 & \ab{1} & \ab{0} & \ab{1} \\
1 & 0 & 1 & \ab{1} & \ab{0} & \ab{0} \\
1 & 1 & 0 & 1 & 1 & 0  \\
1 & 1 & 1 & 1 & 1 & 1 \\
\end{array}
\) & \begin{quantikz}[row sep=small, column sep=small]
\lstick{\(q_0\)} & \octrl{2} & \qw \\
\lstick{\(q_1\)} & \ctrl{1} & \qw \\
\lstick{\(q_2\)} & \targ{}  & \qw
\end{quantikz}
&
\(
\begin{array}{ccc|ccc}
q_0 & q_1 & q_2 & q_0' & q_1' & q_2' \\ \hline
0 & 0 & 0 & 0 & 0 & 0 \\
0 & 0 & 1 & 0 & 0 & 1 \\
0 & 1 & 0 & \ab{0} & \ab{1} & \ab{1} \\
0 & 1 & 1 & \ab{0} & \ab{1} & \ab{0}  \\
1 & 0 & 0 & 1 & 0 & 0 \\
1 & 0 & 1 & 1 & 0 & 1 \\
1 & 1 & 0 & 1 & 1 & 0 \\
1 & 1 & 0  & 1 & 1 & 0  \\
\end{array}
\)\\
\midrule
\begin{quantikz}[row sep=small, column sep=small]
\lstick{\(q_0\)} & \targ{} & \qw \\
\lstick{\(q_1\)} & \ctrl{-1} & \qw \\
\lstick{\(q_2\)} & \ctrl{-2}  & \qw
\end{quantikz}
&
\(
\begin{array}{ccc|ccc}
q_0 & q_1 & q_2 & q_0' & q_1' & q_2' \\ \hline
0 & 0 & 0 & 0 & 0 & 0 \\
0 & 0 & 1 & 0 & 0 & 1 \\
0 & 1 & 0 & 0 & 1 & 0 \\
0 & 1 & 1 & \ab{1} & \ab{1} & \ab{1} \\
1 & 0 & 0 & 1 & 0 & 0 \\
1 & 0 & 1 & 1 & 0 & 1 \\
1 & 1 & 0 & 1 & 1 & 0 \\
1 & 1 & 1 & \ab{0} & \ab{1} & \ab{1} \\
\end{array}
\) & \begin{quantikz}[row sep=small, column sep=small]
\lstick{\(q_0\)} & \targ{} & \qw \\
\lstick{\(q_1\)} & \octrl{-1} & \qw \\
\lstick{\(q_2\)} & \octrl{-2}  & \qw
\end{quantikz}
&
\(
\begin{array}{ccc|ccc}
q_0 & q_1 & q_2 & q_0' & q_1' & q_2' \\ \hline
0 & 0 & 0 & \ab{1} & \ab{0} & \ab{0} \\
0 & 0 & 1 & 0 & 0 & 1\\
0 & 1 & 0 & 0 & 1 & 0 \\
0 & 1 & 1 & 0 & 1 & 1 \\
1 & 0 & 0 & \ab{0} & \ab{0} & \ab{0} \\
1 & 0 & 1 & 1 & 0 & 1 \\
1 & 1 & 0 & 1 & 1 & 0 \\
1 & 1 & 0  & 1 & 1 & 0  \\
\end{array}
\)\\
\midrule
\begin{quantikz}[row sep=small, column sep=small]
\lstick{\(q_0\)} & \targ{} & \qw \\
\lstick{\(q_1\)} & \ctrl{-1} & \qw \\
\lstick{\(q_2\)} & \octrl{-2}  & \qw
\end{quantikz}
&
\(
\begin{array}{ccc|ccc}
q_0 & q_1 & q_2 & q_0' & q_1' & q_2' \\ \hline
0 & 0 & 0 & 0 & 0 & 0 \\
0 & 0 & 1 & 0 & 0 & 1 \\
0 & 1 & 0 & \ab{1} & \ab{1} & \ab{0} \\
0 & 1 & 1 & 0 & 1 & 1 \\
1 & 0 & 0 & 1 & 0 & 0 \\
1 & 0 & 1 & 1 & 0 & 1 \\
1 & 1 & 0 & \ab{0} & \ab{1} & \ab{0} \\
1 & 1 & 1 & 1 & 1 & 1 \\
\end{array}
\) & \begin{quantikz}[row sep=small, column sep=small]
\lstick{\(q_0\)} & \targ{} & \qw \\
\lstick{\(q_1\)} & \octrl{-1} & \qw \\
\lstick{\(q_2\)} & \ctrl{-2}  & \qw
\end{quantikz}
&
\(
\begin{array}{ccc|ccc}
q_0 & q_1 & q_2 & q_0' & q_1' & q_2' \\ \hline
0 & 0 & 0 & 0 & 0 & 0 \\
0 & 0 & 1 & \ab{1} & \ab{0} & \ab{1} \\
0 & 1 & 0 & 0 & 1 & 0 \\
0 & 1 & 1 & 0 & 1 & 1 \\
1 & 0 & 0 & 1 & 0 & 0 \\
1 & 0 & 1 & \ab{0} & \ab{0} & \ab{1} \\
1 & 1 & 0 &  1 & 1 & 0\\
1 & 1 & 1 & 1 & 1 & 1 \\
\end{array}
\)\\
\end{tabular}
\caption{Illustration of the truth table for CCNOT gate with different control and target operations.}
\label{tab:truth_table}
\end{table*}

\section{Proof of~\Cref{cor:kbits_swap}}
\label{app:pf_cor1}
\kbitsswap*
\begin{proof}
    We can prove this by induction on $k$. The base case of $k=0$ is true using~\Cref{lem:toff}. Assume that the statement holds for $k$. Consider, two rows $r_1$ and $r_2$ that differ by $k+1$ bits in their binary representation. WLOG, assume that the bits that differ in  the binary representation of $r_1$ and $r_2$ includes the most significant bit, i.e., 
    \begin{align*}
    r_1 &= \{0\,q_{r_1}(n-1)\dots\,q_{r_1}(0)\} \\ &= \sum_{p=1}^{n} 2^{n-p} q_{r_1}(p), \\
    r_2 &= \{1\,q_{r_2}(n-1)\dots\,q_{r_2}(0)\} \\&= 2^n + \sum_{p=1}^{n} 2^{n-p} q_{r_2}(p),
    \end{align*}
    where $q_{r_1}(p), q_{r_2}(p)\in \{0,1\}$. 
    In order to permute rows $r_1$ and $r_2$: 
    \begin{enumerate}
        \item Permute rows $r_1$ and $r_1' = 2^n + r_1$. This can be done using a single $C^n X$ gate by~\Cref{lem:toff} as they differ by one bit. Now row $r_1$ is in position $r_1'$. 
        \item Permute row in position $r_1'$ and $r_2$. These rows differ by $k$ bits and require $2(k-1)+1\, C^n X$ gates by induction hypothesis. Now original row $r_1$ is in position $r_2$ and original row in position $r_2$ is in position $r_1'$.
        \item Finally, permute rows in position $r_1'$ and $r_1$ again using one $C^n X$ gate. 
    \end{enumerate}
    Thus, we get a total of $1 + 2(k-1)+1 + 1 = 2k+1$ $C^n X$ gates completing the proof.  
\end{proof} 

\bibliographystyle{unsrt}
\bibliography{references}

\begin{thebibliography}{10}

\bibitem{VQE}
Jules Tilly, Hongxiang Chen, Shuxiang Cao, Dario Picozzi, Kanav Setia, Ying Li, Edward Grant, Leonard Wossnig, Ivan Rungger, George~H Booth, et~al.
\newblock The variational quantum eigensolver: a review of methods and best practices.
\newblock {\em Physics Reports}, 986:1--128, 2022.

\bibitem{VQLS}
Carlos Bravo-Prieto, Ryan LaRose, Marco Cerezo, Yigit Subasi, Lukasz Cincio, and Patrick~J Coles.
\newblock Variational quantum linear solver.
\newblock {\em Quantum}, 7:1188, 2023.

\bibitem{berry2019qubitization}
Dominic~W Berry, Craig Gidney, Mario Motta, Jarrod~R McClean, and Ryan Babbush.
\newblock Qubitization of arbitrary basis quantum chemistry leveraging sparsity and low rank factorization.
\newblock {\em Quantum}, 3:208, 2019.

\bibitem{low2019hamiltonian}
Guang~Hao Low and Isaac~L Chuang.
\newblock Hamiltonian simulation by qubitization.
\newblock {\em Quantum}, 3:163, 2019.

\bibitem{QSP}
Yonina~C Eldar and Alan~V Oppenheim.
\newblock Quantum signal processing.
\newblock {\em IEEE Signal Processing Magazine}, 19(6):12--32, 2002.

\bibitem{QSVT}
Andr{\'a}s Gily{\'e}n, Yuan Su, Guang~Hao Low, and Nathan Wiebe.
\newblock Quantum singular value transformation and beyond: exponential improvements for quantum matrix arithmetics.
\newblock In {\em Proceedings of the 51st annual ACM SIGACT symposium on theory of computing}, pages 193--204, 2019.

\bibitem{comm}
Jarrod~R McClean, Jonathan Romero, Ryan Babbush, and Al{\'a}n Aspuru-Guzik.
\newblock The theory of variational hybrid quantum-classical algorithms.
\newblock {\em New Journal of Physics}, 18(2):023023, 2016.

\bibitem{samp}
Nicholas~C Rubin, Ryan Babbush, and Jarrod McClean.
\newblock Application of fermionic marginal constraints to hybrid quantum algorithms.
\newblock {\em New Journal of Physics}, 20(5):053020, 2018.

\bibitem{shad}
Hsin-Yuan Huang, Richard Kueng, and John Preskill.
\newblock Predicting many properties of a quantum system from very few measurements.
\newblock {\em Nature Physics}, 16(10):1050--1057, 2020.

\bibitem{tomo}
Giacomo Torlai, Guglielmo Mazzola, Giuseppe Carleo, and Antonio Mezzacapo.
\newblock Precise measurement of quantum observables with neural-network estimators.
\newblock {\em Physical Review Research}, 2(2):022060, 2020.

\bibitem{liu2021variational}
Hai-Ling Liu, Yu-Sen Wu, Lin-Chun Wan, Shi-Jie Pan, Su-Juan Qin, Fei Gao, and Qiao-Yan Wen.
\newblock Variational quantum algorithm for the poisson equation.
\newblock {\em Physical Review A}, 104(2):022418, 2021.

\bibitem{gnanasekaran2024efficient}
Abeynaya Gnanasekaran and Amit Surana.
\newblock Efficient variational quantum linear solver for structured sparse matrices.
\newblock In {\em 2024 IEEE International Conference on Quantum Computing and Engineering (QCE)}, volume~1, pages 199--210. IEEE, 2024.

\bibitem{nielsen2010quantum}
Michael~A Nielsen and Isaac~L Chuang.
\newblock {\em Quantum computation and quantum information}.
\newblock Cambridge university press, 2010.

\bibitem{lcunum}
Lukas Hantzko, Lennart Binkowski, and Sabhyata Gupta.
\newblock Tensorized pauli decomposition algorithm.
\newblock {\em Physica Scripta}, 99(8):085128, 2024.

\bibitem{liu2021efficient}
Jin-Peng Liu, Herman~{\O}ie Kolden, Hari~K Krovi, Nuno~F Loureiro, Konstantina Trivisa, and Andrew~M Childs.
\newblock Efficient quantum algorithm for dissipative nonlinear differential equations.
\newblock {\em Proceedings of the National Academy of Sciences}, 118(35):e2026805118, 2021.

\bibitem{li2025potential}
Xiangyu Li, Xiaolong Yin, Nathan Wiebe, Jaehun Chun, Gregory~K Schenter, Margaret~S Cheung, and Johannes M{\"u}lmenst{\"a}dt.
\newblock Potential quantum advantage for simulation of fluid dynamics.
\newblock {\em Physical Review Research}, 7(1):013036, 2025.

\bibitem{demirdjian2022variational}
Reuben Demirdjian, Daniel Gunlycke, Carolyn~A Reynolds, James~D Doyle, and Sergio Tafur.
\newblock Variational quantum solutions to the advection--diffusion equation for applications in fluid dynamics.
\newblock {\em Quantum Information Processing}, 21(9):322, 2022.

\bibitem{surana2024efficient}
Amit Surana, Abeynaya Gnanasekaran, and Tuhin Sahai.
\newblock An efficient quantum algorithm for simulating polynomial dynamical systems.
\newblock {\em Quantum Information Processing}, 23(3):105, 2024.

\bibitem{gnanasekaran2024variational}
Abeynaya Gnanasekaran, Amit Surana, and Hongyu Zhu.
\newblock Variational quantum framework for nonlinear pde constrained optimization using carleman linearization.
\newblock {\em Quantum Information \& Computation}, 25(3):260--289, 2025.

\bibitem{demirdjian2025efficient}
Reuben Demirdjian, Thomas Hogancamp, and Daniel Gunlycke.
\newblock An efficient decomposition of the carleman linearized burgers' equation.
\newblock {\em arXiv preprint arXiv:2505.00285}, 2025.

\bibitem{10.5555/2011679.2011682}
M.~Fang, S.~Fenner, F.~Green, S.~Homer, and Y.~Zhang.
\newblock Quantum lower bounds for fanout.
\newblock {\em Quantum Info. Comput.}, 6(1):46–57, jan 2006.

\bibitem{nie2024quantum}
Junhong Nie, Wei Zi, and Xiaoming Sun.
\newblock Quantum circuit for multi-qubit toffoli gate with optimal resource.
\newblock {\em arXiv preprint arXiv:2402.05053}, 2024.

\bibitem{katz2022n}
Or~Katz, Marko Cetina, and Christopher Monroe.
\newblock N-body interactions between trapped ion qubits via spin-dependent squeezing.
\newblock {\em Physical Review Letters}, 129(6):063603, 2022.

\bibitem{lin2022lecturenotesquantumalgorithms}
Lin Lin.
\newblock Lecture notes on quantum algorithms for scientific computation, 2022.

\bibitem{camps2022fable}
Daan Camps and Roel Van~Beeumen.
\newblock Fable: Fast approximate quantum circuits for block-encodings.
\newblock In {\em 2022 IEEE International Conference on Quantum Computing and Engineering (QCE)}, pages 104--113. IEEE, 2022.

\bibitem{camps2024explicit}
Daan Camps, Lin Lin, Roel Van~Beeumen, and Chao Yang.
\newblock Explicit quantum circuits for block encodings of certain sparse matrices.
\newblock {\em SIAM Journal on Matrix Analysis and Applications}, 45(1):801--827, 2024.

\bibitem{sunderhauf2024block}
Christoph S{\"u}nderhauf, Earl Campbell, and Joan Camps.
\newblock Block-encoding structured matrices for data input in quantum computing.
\newblock {\em Quantum}, 8:1226, 2024.

\bibitem{tang_qsvt_lecture}
Ewin Tang.
\newblock Quantum signal processing and singular value transformation: Lecture 1, 2023.

\bibitem{zhang2024circuit}
Xiao-Ming Zhang and Xiao Yuan.
\newblock Circuit complexity of quantum access models for encoding classical data.
\newblock {\em npj Quantum Information}, 10(1):42, 2024.

\bibitem{schaffer1955unitary}
JJ~Sch{\"a}ffer.
\newblock On unitary dilations of contractions.
\newblock In {\em Proc. Amer. Math. Soc}, volume~6, page 322, 1955.

\bibitem{surana2024variational}
Amit Surana and Abeynaya Gnanasekaran.
\newblock Variational quantum framework for partial differential equation constrained optimization.
\newblock {\em arXiv preprint arXiv:2405.16651}, 2024.

\end{thebibliography}
\end{document}